\documentclass[journal]{IEEEtran}
\newcounter{mytempeqncnt}
\newtheorem{thm}{Theorem}
\newtheorem{lem}{Lemma}
\newtheorem{prp}{Proposition}
\newtheorem{crlry}{Corollary}
\newtheorem{defntn}{Definition}
\newtheorem{illexamp}{Illustrative example}

\newenvironment{proof}[1][Proof]{\begin{trivlist}
\item[\hskip \labelsep {\bfseries #1}]}{\end{trivlist}}

\newcommand{\qed}{\nobreak \ifvmode \relax \else
      \ifdim\lastskip<1.5em \hskip-\lastskip
      \hskip1.5em plus0em minus0.5em \fi \nobreak
      \vrule height0.75em width0.5em depth0.25em\fi}

\ifCLASSINFOpdf
  \usepackage[pdftex]{graphicx}
  \graphicspath{{./}{Figs/}}
  \DeclareGraphicsExtensions{.pdf,.jpeg,.png}
\else
\fi
\usepackage{amssymb}
\usepackage{amsmath}
\usepackage{mathtools}
\usepackage{bigdelim}
\usepackage{tikz}
\usepackage{bbm}
\usepackage{mathtools}
\usepackage{centernot}

\usepackage{enumerate}
\usepackage{subcaption}
\usepackage{pgfplots}
\usepackage[caption=false,font=footnotesize]{subfig}
\usepackage{undertilde}
\usepackage{verbatim}
\usetikzlibrary{arrows,calc}
\usepackage{relsize}

\newcommand{\defeq}{\vcentcolon=}

\begin{document}
\title{Self-organizing Networks of Information Gathering Cognitive Agents}

\author{Ahmed~M.~Alaa,~\IEEEmembership{Member,~IEEE}, Kartik Ahuja, and Mihaela van der Schaar,~\IEEEmembership{Fellow,~IEEE}
\thanks{The authors are with the Department of Electrical Engineering, University of California Los Angeles (UCLA), Los Angeles, CA, 90095, USA (e-mail: ahmedmalaa@ucla.edu, ahujak@ucla.edu, mihaela@ee.ucla.edu). This work was funded by the Office of Naval Research (ONR).} 
}
\IEEEspecialpapernotice{(Invited Paper)}
\markboth{}%
{Alaa \MakeLowercase{\textit{et al.}}: Self-organizing Networks of Information Gathering Cognitive Agents}
\maketitle
\begin{abstract}
In many scenarios, networks emerge endogenously as cognitive agents
establish links in order to exchange information. Network formation
has been widely studied in economics, but only on the basis of simplistic
models that assume that the value of each additional piece of information
is constant. In this paper we present a first model and associated
analysis for network formation under the much more realistic assumption
that the value of each additional piece of information depends on
the type of that piece of information and on the information already
possessed: information may be complementary or redundant. We model
the formation of a network as a non-cooperative game in which the
actions are the formation of links and the benefit of forming a link
is the value of the information exchanged minus the cost of forming
the link. We characterize the topologies of the networks emerging
at a Nash equilibrium (NE) of this game and compare the efficiency
of equilibrium networks with the efficiency of centrally designed
networks. To quantify the impact of information redundancy and linking
cost on social information loss, we provide estimates for the Price
of Anarchy (PoA); to quantify the impact on individual information
loss we introduce and provide estimates for a measure we call Maximum
Information Loss (MIL). Finally, we consider the setting in which agents
are not endowed with information, but must produce it. We show that
the validity of the well-known \textquotedblleft law of the few\textquotedblright{}
depends on how information aggregates; in particular, the \textquotedblleft law
of the few\textquotedblright{} fails when information displays complementarities. \end{abstract}
\begin{IEEEkeywords}
Cognitive networking, cognitive agents, information networks, network
formation, self-organizing networks. 
\end{IEEEkeywords}
\IEEEpeerreviewmaketitle{ }
\section{Introduction}
\IEEEPARstart{T}{he} widespread usage of mobile devices, together with the emergence of social-based services and applications, have inspired novel and self-organized networking paradigms that capitalize on the ability of mobile devices to connect and share information in an ad-hoc fashion. Contemporary networks, where users produce and exchange information, are ``socio-technological" in nature; users do not necessarily exploit an exogenously designed network infrastructure, but rather form an endogenous network driven by the individual users' quest for information. In this paper, we present a novel network formation model for information exchange over endogenously formed networks. Albeit being abstract, our model provides insights into understanding and designing many emerging and envisioned classes of applications. 
\subsection{Motivation}
Many emerging networks are formed endogenously by self-interested agents, who take information sharing and production actions. Examples of such networks are: dynamic spectrum management by wireless users \cite{ref1}, social networks overlaid on technological networks \cite{ref4} \cite{ref5}, device-to-device (D2D) communications, vehicular networks \cite{ref7}, Internet-of-Things (IoT) \cite{ref72}, and smart sensor networks \cite{ref73}. In many of these networks, users connect to each other in order to exchange and gather information. For instance, secondary users exchange information about spectrum occupancy in cognitive radio networks \cite{ref8} \cite{ref9}, autonomous rescue robots exchange environmental sensory information \cite{ref72} \cite{ref10}, D2D users engage in short range communications in order to exchange data content of Social Networks Services (SNSs) \cite{refnew1}, and self-interested users take capacity allocation decisions for multicast streaming over networks \cite{refnew2}. Users in such networks possess two key features: they are {\it opportunistic}, in the sense that they exploit their opportunistic encounter with other mobile users to establish short-range communication links with them, and they are {\it cognitive}, in the sense that they need to reason about establishing costly communication links with others given the value of information they can get via these links. Information in this context is an abstraction for any class of data that users gather and process, such as multi-modal content, geographical information, event-related information, cached content, behavioral data, and personal sensory information \cite{refnew3}-\cite{refnew7}. For instance, mobile users who coexist in close proximity can share information about traffic congestion and road accidents which helps them update their routes via applications such as Waze and Google maps, and D2D users can gather offloaded traffic of context-aware applications from other users by forming short-range communications links \cite{refnew1}. Moreover, information can also be produced by the agents themselves in the form of user-generated content, such as the upload and creation of blogs, videos and photos on online social networks (OSN), the purchase of content from service providers in peer-to-peer networks, updating traffic information via an application such as Waze, etc. Thus, users in such networks can jointly decide how much information should they produce, and how much information should they opportunistically acquire from other users. As it is whenever users are self-interested, a game-theoretic framework is naturally deployed to study which networks will emerge at equilibrium and what are their characteristics. Network formation has been studied in the economics, electrical engineering and computer science literature. In the following subsection, we briefly review these related works on endogenous network formation.  

\subsection{Related Works}
Strategic network formation was first studied in the economics literature. Some of this literature \cite{ref12}-\cite{ref141} asks which networks are stable (according to some criteria) and hence more likely to persist and be observed. A (smaller) literature asks which networks emerge as the result of some specific dynamic process \cite{ref15} \cite{ref16}. In all these works, simplistic benefit functions are used: the value of each additional ``good" exchanged is constant \cite{ref12}-\cite{mih2}. However, in realistic settings, information possessed by different agents can be redundant or complementary. For instance, secondary users in a multi-band cognitive radio system may be interested in gathering information about spectrum occupancy for bands that they do not sense by communicating with other users who do sense these bands \cite{ref8};  sensors  deployed over a correlated random field \cite{ref17}-\cite{ref187} may be interested in gathering complementary measurements about some set of physical processes of interest; and mobile users who exchange offloaded traffic of SNSs and context-aware applications are only interested in gathering non-redundant traffic and data updates.

\subsection{Summary of contributions}
This paper introduces a new model for strategic network formation where autonomous cognitive agents exchange valuable information. We refer to such networks as \textit{cognitive information networks} (CIN); networks in which agents self-organize to gather/exchange and produce information about a state of the world. This state of the world can be spectrum occupancy information and primary user activity in a multi-band cognitive radio system, location information provided by anchors of wireless networks, a set of messages sent by information sources in a multicast network, or blogs, videos, and data exchanged by users of social-physical networks. Agents are \textit{cognitive} since they \textit{perceive} information possessed by other agents, \textit{reason} about which links to establish, how much information to produce, and then take information production and link formation \textit{decisions} which result in an endogenously-formed network topology. We assume that agents in a CIN possess different amounts of information, benefit only from gathering non-redundant information, and they form links with each other in order to gather information and maximize their \textit{knowledge} of the state of the world.

Since the information possessed by different agents may be correlated (redundant), and link formation is costly, agents should cognitively select \textit{which} agents to link with. We formulate this problem as a \textit{non-cooperative network formation game}. Using information-theoretic measures for the value of the information possessed by each agent, we aim at characterizing the emerging stable network topologies at Nash Equilibrium (NE). Throughout our analysis, we focus on two classes of linking cost scenarios: homogeneous link formation cost and heterogeneous link formation costs. In the former, connecting to any agent entails the same cost, while in the later, the link cost is recipient-dependent. The link cost can correspond to tokens \cite{tok1} \cite{tok2}, or an abstraction for any monetary, energy, or delay costs incurred by the agent forming the link. An agent in the network is an abstraction for a mobile user, a mobile device, or a transmitter/receiver that is rational and self-interested. 

We show that the networks that emerge at equilibrium are minimally connected; thus, agents tend to minimize the overall cost of constructing the network. With homogeneous link costs, equilibrium leads to a network in which each component is a star. Moreover, we show how information redundancy affects the link cost ranges at which the network becomes connected or disconnected, in addition to its impact on the network efficiency by quantifying the Price-of-Anarchy (PoA). For instance, we show that for networks with low link costs, when the link costs are homogeneous, all emerging networks are efficient; in contrast, information redundancy can induce costly anarchy in networks with heterogeneous link costs.

Finally, we consider a setting in which each agent will not only decide which links to form, but also the amount of information to produce and we provide a characterization for the emerging NE. We show that when the number of agents is large, the fraction of agents producing information at equilibrium depends on the amount of redundancy in the agents' information. When the agents produce strongly correlated information, the fraction of information producers is small and tends to zero as the number of agents tends to infinity: most agents get the information they need from a small set of agents. On the other hand, when agents have uncorrelated information, the number of information producers can grow at the same rate of total number of agents. Thus, such networks violate what Galeotti and Goyal \cite{ref19} call the ``\textit{law of the few}". In addition, we quantify the total amount of information produced in an asymptotically large network and identify scenarios in which the amount of information produced at equilibrium grows with the number of agents.

This paper introduces a new model for cognitive agents exchanging information/knowledge and studying what networks emerge endogenously as a result of self-organizing cognitive agents. Since many applications can use the presented model, we do not delve on the idiosyncratic details of specific applications. The rest of the paper is organized as follows. In Section II, we formalize the network formation game among agents in a CIN. Section III characterizes the emerging stable networks when the link formation costs are homogeneous, and the efficiency of such networks are investigated. Section IV analyzes the network topology and equilibrium efficiency for the case of heterogeneous link costs. The joint information production and link formation game is studied in Section V. Suggested future extensions for our model are provided in Section VI. Finally, conclusions are drawn in Section VII.

\section{Basic Model}

In this section, we discuss the problem setting and propose a basic
model to formulate the endogenous network formation game emerging among cognitive
agents.

\subsection{Information model}

Let $\mathcal{N}=\{1,2,3,...,N\}$ be the set of agents in the CIN.
Each agent $i$ possesses exogenous information in the form of a discrete
random variable $X_{i}$ and aims to form links with other agents
to maximize its utility, which is defined as the benefit from the
total information it possesses minus the linking cost. The formation
of links is costly; thus, an agent has to trade off the benefits of
the information it obtains from another agent versus the cost it needs
to pay for connecting with that agent. The amount of information in
$X_{i}$ is quantified by the \textit{entropy} function $H(X_{i})$.
In addition, the random variables of all agents may be correlated,
which indicates that some agents may possess similar information that
is redundant to that of the other agents. The common information between
agent $i$ and $j$ is captured by the \textit{mutual information}
$I(X_{i};X_{j})$.

The information possessed by the set of agents $\mathcal{N}$ is captured
by an \textit{entropic vector} that we define as follows.

\begin{defntn} \textbf{Entropic vector-} a vector $\overrightarrow{{\bf H}}$
is said to be an entropic vector of order $N$ if there exists a random variable tuple $(X_{1},X_{2},...,X_{N})$,
where associated with any subset $\mathcal{V}$ of $\mathcal{N}$,
there is a joint entropy $H(X_{\mathcal{V}})$ that is an element
of $\overrightarrow{{\bf H}}$, where $X_{\mathcal{V}}=\{X_{i}|i\in\mathcal{V}\}$
\cite{ref20}. \, \IEEEQEDhere
\end{defntn}

The elements of $\overrightarrow{{\bf H}}$ represent the joint entropies between all possible subsets of random variables possessed by agents in $\mathcal{N}$. The set of all entropic vector constitute the \textit{entropic region} which we define as follows.

\begin{defntn} \textbf{Entropic region-}
the entropic region $\Gamma_{N}^{*}\subset \mathbb{R}_{+}^{2^{N}-1}$ is the set of all entropic vectors of order $N$, i.e. the set of all possible entropic vectors that can correspond to the information possessed by $N$ agents. Thus, if a vector $\overrightarrow{{\bf H}}$ is entropic, then $\overrightarrow{{\bf H}}\in\Gamma_{N}^{*}$ \cite{ref20}. 
\, \IEEEQEDhere
\end{defntn}
We denote by $\tilde{\mathcal{H}}$ the set of entropic vectors having $H(X_{1},X_{2},...,X_{N})=\sum_{i=1}^{N}H(X_{i})$,
where $\tilde{\mathcal{H}}\subset\Gamma_{N}^{*}$. The set of entropic
vectors in $\tilde{\mathcal{H}}$ is simply a hyperplane in $\Gamma_{N}^{*}$ that correspond to all entropic vectors with no
information redundancies, which captures the aggregation
models in \cite{ref12} \cite{ref13} \cite{ref141}.

The entropic vector can be constructed as follows. Given the set of
agents $\mathcal{N}$ and a corresponding set of random variables
$\mathcal{X}=\{X_{1},X_{2},...,X_{N}\}$, we construct the set $\mathcal{V}=\mathcal{P}(\mathcal{X})/\{\phi\}$,
where $\mathcal{P}(\mathcal{X})$ is the power set of $\mathcal{X}$.
If $\mathcal{V}=\{v_{1},v_{2},...,v_{|\mathcal{V}|}\}$, then the
entropic vector is given by $\overrightarrow{{\bf H}}=\left(H(X_{v_{i}})\right)_{i=1}^{|\mathcal{V}|}$,
where $|\mathcal{V}|=2^{N}-1$, and $H(X_{v_{i}})$ is the joint entropy
between all random variables in the set $v_{i}$. For instance, if
we have 3 agents in the network, then $\mathcal{V}=\{\{1\},\{2\},\{3\},\{1,2\},\{2,3\},\{1,3\},\{1,2,3\}\}$,
and the entropic vector $\overrightarrow{{\bf H}}$ is given by $\left(H(X_{1}),H(X_{2}),H(X_{3}),H(X_{1,2}),H(X_{1,3}),H(X_{2,3}),\right.$
$\left.H(X_{1,2,3})\right)^{T}$, where $H(X_{1,2})=H(X_{1},X_{2})$.
We denote a single element in the entropic vector as $\overrightarrow{{\bf H}}(v)=H(X_{v})$.
The mutual information between the random variables possessed by any
two subsets $\mathcal{W}$ and $\mathcal{U}$ of agents is given by
\cite{ref21} 
\[
I(X_{\mathcal{W}};X_{\mathcal{U}})=H(X_{\mathcal{W}})+H(X_{\mathcal{U}})-H(X_{\mathcal{W}},X_{\mathcal{U}}).
\]

The total amount of information in the network is given by the joint
entropy of the random variables of individual agents $H(\mathcal{X})=H(X_{1},X_{2},X_{3},...,X_{N}),$
where $H(\mathcal{X})\in\overrightarrow{{\bf H}}$.

The mutual information between any two agents $i$ and $j$ is given
by $I(X_{i};X_{j})=H(X_{i})-H(X_{i}|X_{j})$, where $H(X_{i}|X_{j})$
is the \textit{conditional entropy} which represents the additional
information attained by agent $j$ from connecting to $i$, i.e. the
amount of extra information that $j$ gets when getting the information
of $i$. If this benefit is low, it means that $I(X_{i};X_{j})$ is
high, i.e. $X_{i}$ and $X_{j}$ are highly correlated, and vice versa.
Note that mutual information is symmetric, i.e. $I(X_{i};X_{j})=H(X_{i})-H(X_{i}|X_{j})=H(X_{j})-H(X_{j}|X_{i})$.
Finally we quantify the total amount of redundant information in the network. Let ${\bf p}\left(\mathcal{X}\right) =p(X_{1},X_{2},...,X_{N})$
and ${\bf q}\left(\mathcal{X}\right) =\Pi_{i=1}^{N}p(X_{i})$, where $p(X_{i})$ is the pmf of $X_{i}$. The \textit{Kullback Leibler} (KL) divergence for these
distributions can be computed as follows \cite{ref21} 
\begin{align}
D\left({\bf p}||{\bf q}\right) &=\sum_{\mathcal{X}} {\bf p}\left(\mathcal{X}\right) \log\left(\frac{{\bf p}\left(\mathcal{X}\right)}{{\bf q}\left(\mathcal{X}\right)}\right) \nonumber \\
&=\sum_{i=1}^{N}H(X_{i})-H(X_{1},X_{2},...,X_{N}).
\label{KL22}
\end{align}
The KL divergence is a natural metric for quantifying the distance between probability measures, and it can be obtained in terms of the entropy as shown in (1). In particular, the KL divergence of ${\bf q}\left(\mathcal{X}\right)$ from ${\bf p}\left(\mathcal{X}\right)$ is equal to the difference between the amount of information possessed jointly by the agents, and the corresponding amount of information possessed by the same agents if such information has no redundancies. Throughout the paper, we use $H(X_{-i})$ to denote $H(\mathcal{X}/\{X_{i}\})$, and $\mbox{KL}(\mathcal{X}) = D\left({\bf p}||{\bf q}\right)$ to denote the KL divergence.

\subsection{Network formation game}

Agents benefit from gathering information by linking to other agents.
The link formation strategy adopted by agent $i$ is denoted by a
tuple ${\bf g}_{i}=(g_{ij})_{j\in\{1,...,N\}/\{i\}}\in\{0,1\}^{N-1}$;
$g_{ij}$ = 1 if agent $i$ forms a link with agent $j$ and $g_{ij}$
= 0 otherwise. We assume unilateral link formation where an agent
decides to form a link and solely bears the cost of link formation%
\footnote{Other link formation models, such as link formation with bilateral
consent, can be used with an appropriate solution concept such as
pairwise stability as we discuss in Section VII.}. A strategy profile ${\bf g}$ is defined as the collection of strategies
of all agents, i.e. ${\bf g}\defeq({\bf g}_{i})_{i=1}^{N}\in{\bf G}$,
where ${\bf G}$ is a finite space. When agent $i$ forms a link with
agent $j$, it incurs a cost of $c_{ij}$. We define the topology
of the network as $\mathcal{T}=\{(i,j)\in\mathcal{N}\times\mathcal{N}|\max\{g_{ij},g_{ji}\}=1\}$.
All connected agents exchange information bilaterally; thus $\mathcal{T}$
is an undirected graph. Information is shared between agents that
are indirectly connected and agents do not benefit from receiving
multiple versions of the same information from the same agent. Such
model is suitable for networks with multi-hop relaying where information
is forwarded from one node to another \cite{refz}. We write $i\to j$
to indicate that agent $j$ is reachable by agent $i$ either directly
or indirectly. Define the set of agents that $i$ form links with
(set of neighbors) as $\mathcal{N}_{i}({\bf g})=\{j|g_{ij}=1\}$,
and the set of agents reachable by agent $i$ as $\mathcal{R}_{i}({\bf g})=\{j|i\to j\}$.
Throughout the paper, we adopt the following definitions. 
\begin{defntn} \textbf{Network component-} a component
$\mathcal{C}$ is a set of agents such that $i\to j,\forall i,j\in\mathcal{C}$,
and $i\not\to j,\forall i\in\mathcal{C}$ and $j\notin\mathcal{C}$,
i.e. two agents in two different components cannot share information. \, \IEEEQEDhere 
\end{defntn}
\begin{defntn} \textbf{Minimally connected component-} a component is minimally connected if each agent $i\in\mathcal{C}$
is connected to each agent $j\in\mathcal{C}$ via a unique path. \, \IEEEQEDhere
\end{defntn}
Agents in a component share the information they possess and consequently attain ``informational" benefits that are captured via a utility function. The utility function of agent $i$ is given by 
\begin{equation}
u_{i}({\bf g})=f\left(H(X_{i\cup\mathcal{R}_{i}({\bf g})})\right)-\sum_{j\in\mathcal{N}_{i}({\bf g})}c_{ij},\label{eq02}
\end{equation}
where the function $f(.)$ represents the benefit of agent $i$ from
the information it gathers. We assume that the agents benefit from
acquiring information increases, while the marginal benefit decreases,
with the increase of the amount of information gathered. That is,
in a sensor network setting, the benefit of a sensor node from collecting
information saturates if it is connected to a large number of sensors;
thus, $f(.)$ is assumed to be twice continuously differentiable, increasing, and
concave with $f(0)=0$. Note that the total information acquired by
$i$ in (\ref{eq02}) can be written in terms of the conditional entropies
based on the chain rule as \cite{ref21} 
\[
H(X_{i\cup\mathcal{R}_{i}({\bf g})})=H(X_{i})+\sum_{k=1}^{|\mathcal{R}_{i}({\bf g})|}H(X_{j_{k}}|X_{i},\left\{X_{j_{m}}\right\}_{m=1}^{k-1}),
\]
where $\mathcal{R}_{i}({\bf g})=\{j_{1},j_{2},...,j_{|\mathcal{R}_{i}({\bf g})|}\}$,
which implies that agents benefit by acquiring new information conditioned
on its own information and the information it acquires from other
connections. Moreover, the aggregate information can be expressed
in terms of the mutual information as 
\[
H(X_{i\cup\mathcal{R}_{i}({\bf g})})=H(X_{i})+H(X_{\mathcal{R}_{i}({\bf g})})-I(X_{i};X_{\mathcal{R}_{i}({\bf g})}),
\]
where the term $H(X_{\mathcal{R}_{i}({\bf g})})$ represents the net
information that agent $i$ acquires after connecting to the agents
in $\mathcal{N}_{i}({\bf g})$, where the term $I(X_{i};X_{\mathcal{R}_{i}({\bf g})})$
captures the redundancy between the information of agent $i$ and
the information it acquires from the set $\mathcal{R}_{i}({\bf g})$.
Let ${\bf u}=(u_{1},u_{2},...,u_{N})$. Throughout the paper, we denote
the network formation game by $\mathcal{G}^{N}\langle\mathcal{N},{\bf G},{\bf u},\overrightarrow{{\bf H}}\rangle$.
We assume a complete information scenario, where all agents have knowledge
of the entropic vector $\overrightarrow{{\bf H}}$, the strategy space
${\bf G}$ and the utilities of all agents ${\bf u}$.

\subsection{Stability concept and network efficiency}

The link formation game is formulated as a non-cooperative simultaneous
move game and we focus on the Nash Equilibrium (NE) as the solution concept.
The NE is defined as follows 
\begin{equation}
u_{i}({\bf g}_{i}^{*},{\bf g}_{-i}^{*})\geq u_{i}({\bf g}_{i},{\bf g}_{-i}^{*}),\forall{\bf g}_{i}\in\{0,1\}^{N-1},\forall i\in\mathcal{N},\label{eq6}
\end{equation}
where ${\bf g}_{i}^{*}$ is the NE strategy of agent $i$, and ${\bf g}_{-i}^{*}$ is the NE strategy profile of all users other than
$i$. A strict NE is obtained by making the inequality in (\ref{eq6})
strict. The game can have multiple NE defined as ${\bf G^{*}}=\{{\bf g}^{*}|\,\,\forall u_{i}({\bf g}_{i}^{*},{\bf g}_{-i}^{*})\geq u_{i}({\bf g}_{i},{\bf g}_{-i}^{*}),\forall{\bf g}_{i}\in\{0,1\}^{N-1}\}$.
In the following Theorem, we show that there exists at least one network satisfying the NE conditions, i.e. ${\bf G^{*}}\neq\phi$.
\begin{thm}
{\it (The Existence of Nash Equilibrium)} A pure strategy NE always exists
for $\mathcal{G}^{N}=\langle\mathcal{N},{\bf G},{\bf u},\overrightarrow{{\bf H}}\rangle$. 
\end{thm}
\begin{proof} See Appendix A. \, \IEEEQEDhere \end{proof}

The social welfare of the network formation game is defined as the
sum of agents' individual utilities. For a strategy profile ${\bf g}$,
the social welfare is defined as 
\begin{equation}
U({\bf g})\defeq\sum_{i\in\mathcal{N}}u_{i}({\bf g}).\label{eq7}
\end{equation}
A strategy profile ${\bf \tilde{g}}$ is called \textit{socially optimal}
if it maximizes the social welfare (achieves the social optimum $\tilde{U}$),
i.e. 
\begin{equation}
\tilde{U}\defeq U({\bf \tilde{g}})\geq U({\bf g}),\forall{\bf g}\in{\bf G}.\label{eq8}
\end{equation}
When there are multiple equilibria, we use two metrics to assess the
equilibrium efficiency. First, we adopt the \textit{Price of Anarchy}
(PoA) to quantify the impact of the agents' selfish behavior on the
social welfare. The PoA is defined as the ratio between the social
optimum and the lowest social welfare achieved at equilibrium, i.e.
\begin{equation}
\mbox{PoA}=\frac{\tilde{U}}{\min_{{\bf g^{*}}\in{\bf G^{*}}}U({\bf g^{*}})}.\label{eq9}
\end{equation}

In addition, we analyze the impact of the agents selfish behavior
on the information gathering process by defining a novel metric that
we term the \textit{Maximum Information Loss} (MIL). The MIL is defined
as the maximum difference between the amount of information gathered
by any agent at two different equilibria as shown in (\ref{regn1}).
Unlike the PoA, the MIL quantifies the maximum information loss without
considering the link cost. In addition, while the PoA considers the
welfare of \textit{all agents}, the MIL quantifies the highest information
loss incurred by an \textit{agent} in the worst case.

\begin{figure*}[!t]
\setcounter{mytempeqncnt}{\value{equation}} \setcounter{equation}{6}
\begin{equation}
\mbox{MIL}=\max_{i}\left(\sup_{{\bf g}_{u}^{*}\in G^{*}}H(X_{i}\cup X_{\mathcal{R}_{i}({\bf g}_{u}^{*})})-\inf_{{\bf g}_{v}^{*}\in G^{*}}H(X_{i}\cup X_{\mathcal{R}_{i}({\bf g}_{v}^{*})})\right).\label{regn1}
\end{equation}
\setcounter{equation}{\value{mytempeqncnt}+1} \hrulefill{}\vspace*{4pt}
 
\end{figure*}

\section{Nash Equilibrium Analysis for Homogeneous Link Costs}

In this section, we assume that the cost of forming a link between
any two agents $i$ and $j$ is given by $c_{ij}=c,\forall\, i,j\in\mathcal{N}$.
The goal of this section is to answer the following question: \textit{given
an entropic vector $\overrightarrow{{\bf H}}$, what are the network
topologies $\mathcal{T}$ that can emerge at an NE of the game $\mathcal{G}^{N}$
when the link costs are homogeneous?} We start with the following
motivating example to identify different factors that affect the equilibria
of $\mathcal{G}^{2}$. \setcounter{subsection}{0}

\subsection{Motivating example for two-agents interaction: does information redundancy
matter?}

Consider a simple network with only two agents ($N=2$) possessing
random variables $X_{1}$ and $X_{2}$. We aim at characterizing the
equilibria of $\mathcal{G}^{2}=\langle\{1,2\},{\bf G},{\bf u},\overrightarrow{{\bf H}}\rangle$.
The strategy of agent 1 is simply a linking decision $g_{12}\in\{0,1\}$,
while for agent 2, the strategy is $g_{21}\in\{0,1\}$. We write $\mathcal{G}^{2}$
in normal form in Table \ref{tab:title}, where the row player is
agent 2 and the column player is agent 1. Each cell displays the utilities
of agents 1 and 2 respectively. Assume that the link cost is the same
for both agents and equal to $c$. It can be easily shown that the
payoffs of agent 1 are given by $u_{1}(g_{12}=1,g_{21}=1)=u_{1}(g_{12}=1,g_{21}=0)=f\left(H(X_{1},X_{2})\right)-c$,
$u_{1}(g_{12}=0,g_{21}=1)=f\left(H(X_{1},X_{2})\right)$, and $u_{1}(g_{12}=0,g_{21}=0)=f\left(H(X_{1})\right)$.

\begin{table}[!h]
\captionsetup{font= small}
\caption{Two agent network formation game in normal form}
\label{tab:title}
\begin{center}
\centering
\begin{tabular}{ r|p{3cm}| p{3cm} |}
\multicolumn{1}{r}{}
 &  \multicolumn{1}{c}{$g_{12} = 1$}
 & \multicolumn{1}{c}{$g_{12} = 0$} \\
\cline{2-3}
$g_{21} = 1$ &  $u_{1}(g_{12} = 1, g_{21} = 1),$ \newline $u_{2}(g_{12} = 1, g_{21} = 1)$  & $u_{1}(g_{12} = 0, g_{21} = 1),$ \newline $u_{2}(g_{12} = 0, g_{21} = 1)$ \\
\cline{2-3}
$g_{21} = 0$ & $u_{1}(g_{12} = 1, g_{21} = 0),$ \newline $u_{2}(g_{12} = 1, g_{21} = 1)$ & $u_{1}(g_{12} = 0, g_{21} = 0),$ \newline $u_{2}(g_{12} = 0, g_{21} = 0)$ \\
\cline{2-3}
\end{tabular} 
\end{center}
\end {table}

Fig. \ref{figM1} depicts the entropic region $\Gamma_{2}^{*}$ of the two random variables $X_{1}$ and $X_{2}$. The entropic region $\Gamma_{2}^{*}$ can be easily constructed by applying the three Shannon inequalities $H(X_{1})\leq H(X_{1},X_{2})$, $H(X_{2})\leq H(X_{1},X_{2})$, and $H(X_{1})+H(X_{2})\geq H(X_{1},X_{2})$. The intersection of these three hyperplanes in $\mathbb{R}_{+}^{3}$ results in the polyhedral cone depicted in Fig. \ref{figM1}. The distance between an entropic vector (depicted by a thick dot inside $\Gamma_{2}^{*}$) and the corresponding entropic vector on $\tilde{\mathcal{H}}$  the light-colored hyperplane) with the same $H(X_{1})$ and $H(X_{2})$ is equal to the KL divergence. If $\mbox{KL}(X_{1},X_{2})=0$, then the entropic vector lies on $\tilde{\mathcal{H}}$, and the 2 agents have non-redundant information.

\begin{figure}[!h]
\centering
\begin{tikzpicture}[join=round]
    \tikzstyle{conefill} = [fill=blue!20,fill opacity=0.8]
    \tikzstyle{ann} = [fill=white,font=\footnotesize,inner sep=1pt]
    \tikzstyle{ghostfill} = [fill=white]
         \tikzstyle{ghostdraw} = [draw=black!50]			
	\draw[arrows=->,line width=.4pt](0,0,0)--(0,2.5,0);
	\draw[arrows=->,line width=.4pt](0,0,0)--(2.5,0,0);
	\draw[arrows=->,line width=.4pt](0,0,0)--(0,0,2.5);
		\tikzstyle{conefill} = [fill=red!20,fill opacity=0.8]	
    \filldraw[conefill](0,0,0)--(0,2.5,2.5)--(2.5,3.5,2.5)
                        --cycle;
		\tikzstyle{conefill} = [fill=blue!20,fill opacity=0.8]
    \filldraw[conefill](0,0,0)--(2,2,2)--(2.5,3.5,2.5)
                        --cycle;
    \draw[thick,arrows=<->](1.4,2.5,1)--(1.4,1.85,1);
		\draw [dashed] (1.4,1.85,1)--(1.4,-0.5,1);
		\draw [dashed] (-0.5,-0.5,0.9)--(1.4,-0.5,0.9);
		\draw [dashed] (1.4,-0.5,-1.2)--(1.4,-0.5,0.9);
		\fill(1.4,1.8,1) circle (2pt);
		\tikzstyle{conefill} = [fill=blue!20,fill opacity=0.8]
    \filldraw[conefill](0,0,0)--(0,2.5,2.5)--(2,2,2)
                        --cycle;
    |
		\draw[dashed, arrows=<-,line width=.4pt](1.52,2.2,1.2)--(2.25,2.2,-0.5);
    \node[ann] at (3.25,2.2,-0.5) {$\mbox{KL}(X_{1},X_{2})$};
    \node[ann] at (-0.5,2.5,2.5) {$\Gamma_{2}^{*}$};
    \path (-0.5,0.5,2.5) node[ann] {$H(X_{1})$}
        (0.5,3,0.5) node[ann] {$H(X_{1},X_{2})$}
        (3,0.5,0.5) node[ann] {$H(X_{2})$};
				
				\draw(0,0,0)--(0,3.5,3.5);
				\draw(0,0,0)--(2,2,0);
				\draw(2.5,3.5,2.5)--(3,4.25,2.5);

\end{tikzpicture}
\captionsetup{font= small}
\caption{The entropic region $\Gamma_{2}^{*}$ for 2 random variables.} 
\label{figM1}
\end{figure}
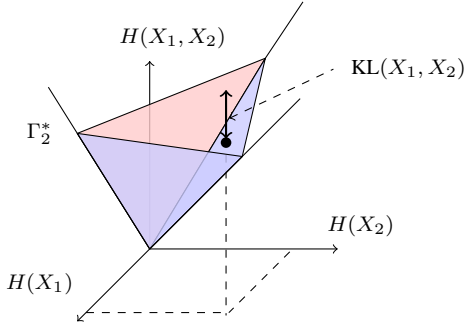

The equilibria of this game depend on both the link cost and the entropic
vector, which corresponds to the amount of information redundancy. For an arbitrary entropic vector, the game has two possible equilibria ${\bf g}^{*}=(g_{12}=1,g_{21}=0)$
and ${\bf g}^{*}=(g_{12}=0,g_{21}=1)$ if $c\leq f(H(X_{1},X_{2}))-f\left(\max\{H(X_{1}),H(X_{2})\}\right)$.
Assume that $H(X_{1})>H(X_{2})$. Therefore, the network has a unique
equilibrium ${\bf g}^{*}=(g_{12}=0,g_{21}=1)$ when $f(H(X_{1},X_{2}))-f(H(X_{1}))\leq c\leq f(H(X_{1},X_{2}))-f(H(X_{2}))$,
and a unique equilibrium ${\bf g}^{*}=(g_{12}=0,g_{21}=0)$ when $c\geq f(H(X_{1},X_{2}))-f(H(X_{2}))$.
On the other hand, if we fix the link cost and the entropies $H(X_{1})$
and $H(X_{2})$, we observe that the equilibria change by changing
the KL divergence. For instance, the network has two equilibria ${\bf g}^{*}=(g_{12}=1,g_{21}=0)$
and ${\bf g}^{*}=(g_{12}=0,g_{21}=1)$ when $c\leq f(H(X_{1})+H(X_{2})-\mbox{KL}(\mathcal{X}))-f\left(\max\{H(X_{1}),H(X_{2})\}\right)$.
Thus, as the entropic vector becomes closer to the hyperplane $\tilde{\mathcal{H}}$,
i.e. $\mbox{KL}(\mathcal{X})$ decreases, the cost threshold for which
these two equilibria emerge increases. This means that the characterization
of the NE is sensitive to the amount of information redundancy $\mbox{KL}(\mathcal{X})$,
even if we fix the individual entropies $H(X_{1})$ and $H(X_{2})$.
Note that the strategy profile ${\bf g}=(g_{12}=1,g_{21}=1)$ never
emerges as an NE since under such profile any of the two agents can break the link formed and get a strictly higher utility. 

\subsection{Characterization of the NE for $\mathcal{G}^{N}$}
In this subsection, we present a generic characterization for the NE of $\mathcal{G}^{N}$.
\begin{prp} (Network minimality) \textit{In every NE, all
network components are minimally connected.} \end{prp}

\begin{proof} See Appendix B. \, \IEEEQEDhere \end{proof}

Proposition 1 implies that agents in each component will form the
minimal number of links possible to gather the maximum amount of information. This results from indirect information sharing within each
network component, i.e. if there exists a path to an agent then there
is no extra benefit in making a direct link to that agent since all
the information from that agent is already accessible.

Next, we characterize the connectivity of the network as a function
of the link cost in the following Lemma. 

\begin{lem} {\it (Network connectivity regions)}
\begin{enumerate} [(i)]
\item If $c\leq c_{l}$, with $c_{l}=f\left(H(\mathcal{X})\right)-f(\min_{i}H(X_{-i}))$, then, at every NE (a) the network is minimally connected (the network has one component) and (b) the amount of information possessed by each agent is $H(\mathcal{X})$ (all information is shared). 
\item If $c\geq c_{u}$, where $c_{u}=f\left(H(\mathcal{X})\right)-f(\min_{i}H(X_{i}))$, then there is a unique NE which is strict. At this equilibrium, the network is fully disconnected and the amount of information possessed by each agent $i$ is $H(X_{i})$ (no information is shared).
\end{enumerate}
\end{lem}
\begin{proof} See Appendix C. \, \IEEEQEDhere \end{proof}

From the above Lemma, we can see that three factors affect the connectivity
of a network: the link cost, the amount of information possessed by each agent, and
the redundancies among the agents' information. Based on the result of Lemma 1, we define three regions for the  connectivity of the NE networks based on the link cost as follows: 
\begin{itemize}
\item \textit{Connected agents region} ($\mathcal{K}_{C}$): A network with an entropic vector $\overrightarrow{{\bf H}}$ has a single component when the link cost is $c\leq c_{l}$. 
\item \textit{Isolated agents region} ($\mathcal{K}_{I}$): The network
has $N$ components when the link cost is $c\geq c_{u}$. 
\item \textit{Mixed region} ($\mathcal{K}_{M}$): Depending on the entropic
vector, the network can have different number of components ranging
from 1 to $N$ when the link cost is $c_{l}\leq c\leq c_{u}$. 
\end{itemize}
While the connectivity regions describe the impact of link cost
on network topology, they also have informational significance. For
instance, the amount of information possessed by any agent in the
$\mathcal{K}_{C}$ region is $H(\mathcal{X})$, while in the $\mathcal{K}_{I}$
region, no agent $i$ gathers any extra information other than its
own intrinsic information $H(X_{i})$. On the other hand, agents in
the $\mathcal{K}_{M}$ region can end up gathering different amounts
of information as there are potentially multiple equilibria with different
topologies and connectedness. In the following illustrative example, we demonstrate the impact of the link cost and information redundancy on the network's connectivity regions. 



\begin{illexamp} To illustrate the impact of information redundancy and link cost on the NE networks' connectivity, we plot the $\mathcal{K}_{M}$, $\mathcal{K}_{C}$, and $\mathcal{K}_{I}$ regions in the link cost-information redundancy plane for 2 different families of entropic vectors. Assume that
we have a 3-agent CIN, with $H(X_{1})>H(X_{2})$, and $H(X_{2})=H(X_{3})$, and that agent 1 has non-redundant information, i.e. the random variable $X_{1}$ is independent on $X_{2}$ and $X_{3}$. Thus, we have $\mbox{KL}(\mathcal{X})=I(X_{2};X_{3})$. We consider two different families of entropic vectors (i.e. two different assignments for the values of individual agents' entropies), the first is given by $(H^{1}(X_{1}) = 5, H^{1}(X_{2}) = 4, H^{1}(X_{3}) = 4)$, whereas the second is given by $(H^{2}(X_{1}) = 7, H^{2}(X_{2}) = 4, H^{2}(X_{3}) = 2)$. The connectivity regions associated with entropic vector family $i$ is denoted by $\left(\mathcal{K}_{C}^{i},\mathcal{K}_{I}^{i},\mathcal{K}_{M}^{i}\right)$. An exemplary utility function of $f(x)=\log(1+x)$ is used. In Fig. 2, we plot the connectivity regions in the cost-KL divergence plane for the 2 families of entropic vectors. For both families of entropic vectors, the $\mathcal{K}_{M}$ region shrinks as the information redundancy increases. That is, when agents share more information in common, the NE network connectivity becomes less ``uncertain" since the $\mathcal{K}_{M}$ region (which is the only region with potentially multiple equilibria with different levels of connectivity) in this case will correspond to a limited range of link costs. Moreover, we note that for the first family of entropic vectors, when agents 2 and 3 information are fully redundant (i.e. $\mbox{KL}(\mathcal{X}) = 4$), we have a sharp threshold on the link cost, below which we have a connected network, and above which we have a fully disconnected network (i.e. the $\mathcal{K}_{M}$ region is empty). The intuition behind this is that since agents 2 and 3 are fully ``correlated", they only benefit from connecting to agent 1. Thus, agent 1 acts as the only information source, and it is the benefit from getting agent's 1 information that solely determines the cost at which the network would be connected or not. If agents 2 and 3 information are not redundant, they add value to the network, and the cost thresholds become dependent on their information as well. However, for the second family of entropic vector, since there is more heterogeneity in the amount of information possessed by the agents, no single agents monopolizes the information at any value of the KL divergence, thus the $\mathcal{K}_{M}$ region does not vanish for the second
vector for any value of $\mbox{KL}(\mathcal{X})$. \, \IEEEQEDhere \end{illexamp}

While Lemma 1 focuses on the impact of link cost on the connectivity of the network, it does not provide a complete characterization for an NE network. In the next Theorem, we give the necessary and sufficient conditions for the emergence of an arbitrary CIN topology in NE.

\begin{figure}[!t]
\centering
\includegraphics[width=3in]{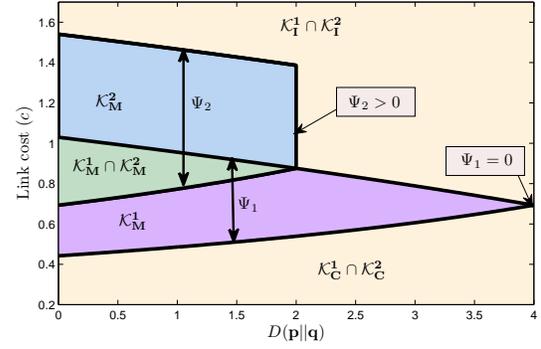}
\captionsetup{font= small}
\caption{Impact of link cost and information redundancy on the network's connectivity.}
\label{fig_pase3}
\end{figure}

\begin{thm}
A network in which the components are precisely $\{\mathcal{C}_{1},\mathcal{C}_{2},...,\mathcal{C}_{K}\}$
can be supported in a NE if and only if the following relationships
between the cost and the value of information are satisfied 
\begin{enumerate}
\item $f(H(X_{\mathcal{C}_{i}}))-\min\{f(H(X_{\mathcal{C}_{i}/\{j\}})),f(H(X_{j}))\}\geq c$,
$\forall i\in\{1,2,...,K\},j\in\mathcal{C}_{i}$. 
\item $f(H(X_{\mathcal{C}_{i}\cup\mathcal{C}_{j}}))-f(H(X_{\mathcal{C}_{i}}))\leq c,\,\forall i,j\in\{1,2,...,K\}$. 
\end{enumerate}
\end{thm}
\begin{proof} See Appendix D. \, \IEEEQEDhere \end{proof}

From Theorem 2 we know that, at NE, the network is generally composed of multiple components and each component is minimally connected. Each component possesses a set of random variables that are jointly highly correlated to the joint random variables possessed by other components. Condition (1) in Theorem 2 implies that each agent in a component either benefits from forming a link to some other agent in that component, or other agents benefit from linking to it, while condition (2) implies that agents in different components have no incentives to connect to agents in other components. Note that due to indirect information sharing, many equilibria can exist with highly variant topologies. In the subsequent Theorem, we refine the equilibrium notion used, and we determine the specific topologies emerging in a strict NE.
\begin{thm}
A network is a strict NE if and only if the following conditions are
simultaneously satisfied 
\begin{itemize}
\item All conditions stated in Theorem 2 are satisfied. 
\item For each component $\mathcal{C}$ of size $M>1$, there exists a set
$\zeta\subseteq\mathcal{C}$ with $|\zeta|\geq M-1$ such that 
\[
\zeta=\{j\,|\, f(H(X_{\mathcal{C}}))-f(H(X_{\mathcal{C}/\{j\}}))>c\}.
\]

\item Each non-singleton component forms a core-sponsored star topology, where the periphery agents belong to the set $\zeta$. 
\end{itemize}
\end{thm}
\begin{proof} See Appendix E. \,\IEEEQEDhere \end{proof}

This Theorem states that for homogeneous link formation costs, each
network component of size $M$ comprises a single agent bearing the
cost of getting connected to $M-1$ other agents. Such networks exhibit
a \textit{core-periphery} structure, i.e. a single agent at the core
is connected to a set of $M-1$ periphery agents. The conditions in
Theorem 3 state that the periphery agents must be \textit{high entropy}
agents. This is because the benefit obtained by connecting to a periphery
agent $j$ at equilibrium must exceed the cost, i.e. $f(H(X_{\mathcal{C}}))-f(H(X_{\mathcal{C}/\{j\}}))>c$.
The intuition behind this condition is as follows. For an agent to
be a periphery agent, it must have both high entropy and low redundancy
with the information possessed by other component members such that
core agents have an incentive to form a link with it. Fig. \ref{fig_pase12}
depicts an exemplary topology of a CIN at strict NE for various link
formation cost ranges. 

In the next subsection, we study the efficiency
of the NE networks and compare the self-organized CINs to those designed
by a network planner.

\begin{figure}[!t]
\centering
\includegraphics[width=3.5in]{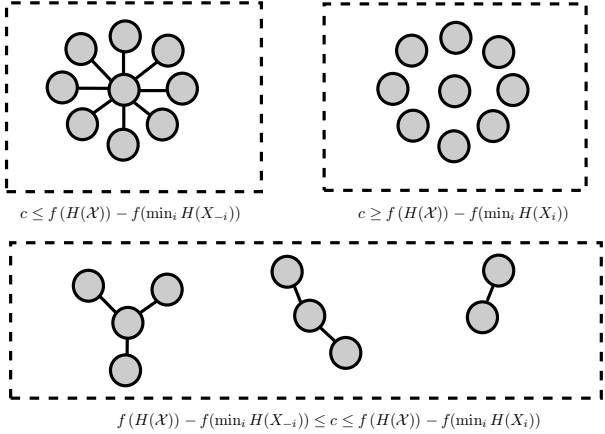}
\captionsetup{font= small}
\caption{The exemplary strict NE topologies for various link cost ranges.}
\label{fig_pase12}
\end{figure}

\subsection{Equilibrium efficiency analysis}

The goal of this subsection is to investigate the equilibrium efficiency of $\mathcal{G}^{N}$ with homogeneous link costs by quantifying the PoA and the MIL. We start by quantifying the PoA of CINs in the following Lemma.

\begin{lem}
For a CIN with homogeneous link costs, the Price-of-Anarchy satisfies
\[
\mbox{PoA}=1,\,\forall\left(\overrightarrow{{\bf H}},c\right)\in\mathcal{K}_{C}\cup\mathcal{K}_{I},
\]
and 
\[
\mbox{PoA} < \frac{Nf(H(\mathcal{X}))}{\sum_{i=1}^{N}f(H(X_{i}))},\,\forall\left(\overrightarrow{{\bf H}},c\right)\in\mathcal{K}_{M}.
\]
\end{lem} 
\begin{proof}  
See Appendix F. \, \IEEEQEDhere
\end{proof}
This Lemma shows that all NE networks in the $\mathcal{K}_{C}$ and $\mathcal{K}_{I}$ regions are socially optimal. While in the $\mathcal{K}_{C}$ region multiple equilibria exist, they all have the same social welfare of $N f(H(\mathcal{X})) - (N-1) c$. However, in the $\mathcal{K}_{M}$ region the NE networks may not be socially optimal, and we give an upper bound on the PoA. When all agents possess non-redundant information, the PoA is upper bounded by $N$, whereas when agents possess redundant information, we have $\mbox{PoA} < \frac{N f(\max_{i}H(X_{i}))}{\sum_{i=1}^{N}f(H(X_{i}))} < N$, which gives an indication that information redundancy reduces the PoA in the $\mathcal{K}_{M}$ region\footnote{This is intuitive since when information redundancy increases, the socially optimal welfare decreases, while the welfare of a disconnected network is fixed, which means that the PoA decreases.}. While the social welfare captures the sum utilities, it does not quantify the individual losses by agents. In the next corollary, we quantify the MIL for different connectivity regions.

\begin{crlry} For a CIN with homogeneous link cost, the MIL satisfies
\[
\mbox{MIL}=0,\,\forall\left(\overrightarrow{{\bf H}},c\right)\in\mathcal{K}_{C}\cup\mathcal{K}_{I},
\]
and 
\[
\mbox{MIL}\leq H(\mathcal{X})-\min_{i}H(X_{i}),\,\forall\left(\overrightarrow{{\bf H}},c\right)\in\mathcal{K}_{M}.
\]
\end{crlry} \begin{proof} See Appendix G. \, \IEEEQEDhere
\end{proof}

Fig. X depicts the PoA for a 3-agent CIN with the first family of entropic vectors defined in illustrative example 1. It can be seen that the PoA is greater than 1 only in the $\mathcal{K}_{M}$ region. In addition, the PoA decreases as the KL divergence increases, since the value of information in the network decreases, which means that the best equilibrium (connected network) achieves a smaller social welfare while the welfare of the worst equilibrium (fully disconnected network) is independent of the KL divergence. The PoA also decreases as the link cost increases. From Fig. X, we can see that when $\mbox{KL}(\mathcal{X})=4$, the network exhibit an empty $\mathcal{K}_{M}$ region, i.e. the network changes from a connected to a fully disconnected network if the cost exceeds a certain threshold. Thus, for $\mbox{KL}(\mathcal{X})=4$ the network is robust to efficiency loss for all values of link cost as the $\mathcal{K}_{M}$ region is the only region where efficiency loss can occur. Fig. X depicts the MIL upper bound for the same network. It is also observed that the MIL upper bound decreases monotonically with the increasing information redundancy.

\section{Nash Equilibrium Analysis for Heterogeneous Link Costs}

In this section, we extend the analysis done in the previous section
for the game $\mathcal{G}^{N}$, but assuming that the cost of link
formation is exclusively recipient-dependent, i.e. $c_{ij}=c_{j},\,\forall\, i$. It is easy to show that Proposition 1 applies to the case of heterogeneous link costs, i.e. all network components that satisfy the NE conditions are minimally connected.

\subsection{Characterization of the NE for $\mathcal{G}^{N}$}

The following proposition relates the link costs to the connectivity
of the NE networks.

\begin{prp}
\begin{enumerate} [(i)]
\item If $c_{i}<f(H(\mathcal{X}))-f(H(X_{-i})), \forall i \in \mathcal{N}$, then, at every NE (a) the network is minimally connected (the network has one component) and (b) the amount of information possessed by each agent is $H(\mathcal{X})$ (all information is shared). 
\item If $f(H(\mathcal{X}))-f(\min_{j}H(X_{-j}))<\min_{k\in\mathcal{N}/\{i\}}c_{k},$ where $i=\arg\min_{j}H(X_{-j})$, then there is a unique NE which is strict. At this equilibrium, the network is fully disconnected and the amount of information possessed by each agent $i$ is $H(X_{i})$ (no information is shared).
\end{enumerate}
\end{prp}
\begin{proof} This can be proven straightforwardly using
the same arguments in the proof of Lemma 1. \, \IEEEQEDhere \end{proof}

This proposition shows that the network topology is highly dependent
on the heterogeneity of the agents as it depends both on the heterogeneous
costs and heterogeneous information of agents. Also the case when
all NE networks are connected corresponds to the $\mathcal{K}_{C}$
region in the homogeneous cost scenario, while the case when the NE
is a fully disconnected network corresponds to the $\mathcal{K}_{I}$
region. An appropriate definition for the connectivity regions for
the heterogeneous cost case is given by (\ref{Moddef1}), (\ref{Moddef2}),
and (\ref{Moddef3}). 
\begin{figure*}[!t]
\setcounter{mytempeqncnt}{\value{equation}} \setcounter{equation}{8}
\begin{equation}
\mathcal{K_{C}}=\left\{ \left(\overrightarrow{{\bf H}},{\bf c}=(c_{1},c_{2},...,c_{N})\right)\left|\,{\bf c}\in\mathbb{R}_{N}^{+},\overrightarrow{{\bf H}}\in\Gamma_{N}^{*},\,\mbox{and}\, c_{i}<f(H(\mathcal{X}))-f(H(X_{-i})),i=\arg\min_{j}H(X_{-j})\right.\right\} .\label{Moddef1}
\end{equation}
\begin{equation}
\mathcal{K}_{I}=\left\{ \left(\overrightarrow{{\bf H}},{\bf c}=(c_{1},c_{2},...,c_{N})\right)\left|\,{\bf c}\in\mathbb{R}_{N}^{+},\overrightarrow{{\bf H}}\in\Gamma_{N}^{*},\,\mbox{and}\, f(H(\mathcal{X}))-f(\min_{j}H(X_{-j}))<\min_{k\in\mathcal{N}/\{i\}}c_{k},i=\arg\min_{j}H(X_{-j})\right.\right\} .\label{Moddef2}
\end{equation}
\begin{equation}
\mathcal{K}_{M}=\left\{ \left(\overrightarrow{{\bf H}},{\bf c}=(c_{1},c_{2},...,c_{N})\right)\left|\forall\left(\overrightarrow{{\bf H}},{\bf c}\right)\notin\mathcal{K}_{C}\cup\mathcal{K_{I}},{\bf c}\in\mathbb{R}_{N}^{+},\overrightarrow{{\bf H}}\in\Gamma_{N}^{*}\right.\right\} .\label{Moddef3}
\end{equation}
\setcounter{equation}{\value{mytempeqncnt}+3} \hrulefill{}\vspace*{4pt}
 
\end{figure*}

In the following Theorem, we give a generic characterization for this class of networks in NE.

\begin{thm}
A network in which the components are precisely $\{\mathcal{C}_{1},\mathcal{C}_{2},...,\mathcal{C}_{K}\}$
can be supported in a NE if and only if the following relationships between the cost and the value of information are satisfied 
\begin{enumerate}
\item $f(H(X_{\mathcal{C}_{i}\cup\mathcal{C}_{j}}))-f(H(X_{\mathcal{C}_{i}}))\geq\min_{k\in\mathcal{C}_{j}}c_{k},\,\forall i,j\in\{1,2,...,K\}$. 
\item $f(H(X_{\mathcal{C}_{i}}))\geq\min\{f(H(X_{\mathcal{C}_{i}/\{j\}}))+c_{j},f(H(X_{j}))+\min_{k\in\mathcal{C}_{i}/\{j\}}c_{k}\}$,
$\forall i\in\{1,2,...,K\},j\in\mathcal{C}_{i}$. 
\end{enumerate}
\end{thm}
\begin{proof} This can be proven following the same idea for the
proof of Theorem 2. \, \IEEEQEDhere \end{proof}

Note that unlike the homogeneous cost scenario, we cannot characterize and plot the connectivity versus a single value for link cost since the link cost is now a multidimensional parameter. In the next subsection, we analyze the efficiency of the NE networks.

\subsection{Equilibrium Efficiency Analysis}

In this subsection, we quantify the impact of the link costs heterogeneity on the network efficiency. Unlike the case of the homogeneous link costs, we show that information redundancy induces costly anarchy in the $\mathcal{K}_{C}$ region when the link costs are recipient-dependent. In the following Lemma, we quantify the PoA for the $\mathcal{K}_{C}$ and $\mathcal{K}_{I}$ regions.
 
\begin{lem}
For a CIN with heterogeneous link costs, the PoA satisfies 
\[
\mbox{PoA}=\left\{ \begin{array}{lr}
1, & :\forall\left(\overrightarrow{{\bf H}},{\bf c}\right)\in\mathcal{K}_{I}\\
\frac{Nf(H(\mathcal{X}))-(N-1)\min_{k}c_{k}}{Nf(H(\mathcal{X}))-\sum_{j=1}^{N}c_{j}+\min_{k}c_{k}}, & :\forall\left(\overrightarrow{{\bf H}},{\bf c}\right)\in\mathcal{K}_{C}
\end{array}\right.
\]
and 
\[
\mbox{PoA} < \frac{Nf(H(\mathcal{X}))}{\sum_{i=1}^{N}f\left(H(X_{i})\right)}:\forall\left(\overrightarrow{{\bf H}},{\bf c}\right)\in\mathcal{K}_{M}.
\] 
\end{lem}
\begin{proof} See Appendix H. \, \IEEEQEDhere \end{proof}

Thus, unlike in the homogeneous cost scenario, not all NE networks in the $\mathcal{K}_{C}$ region are socially optimal. In fact, any NE network other than a \textit{periphery-sponsored} star with the agent having the lowest link cost residing in the core, is not socially optimal. How does information redundancy affect the PoA in such networks? The following Theorem answers this question.
\begin{thm}
For a CIN with recipient-dependent link costs in the $\mathcal{K}_{C}$ region and for fixed values of the individual agents' entropies, the Price-of-Anarchy is a monotonically increasing function of the total information redundancy. 
\end{thm}
\begin{proof} See Appendix I. \, \IEEEQEDhere \end{proof}

Thus, in stark contrast with the results obtained for the homogeneous cost CINs, Theorem 5 states that information redundancy induces costly anarchy for a network in $\mathcal{K}_{C}$ region. This results from the heterogeneity of the link formation costs, which promotes anarchy in the network as agents are no longer indifferent to the links they form as in the homogeneous cost scenario. As a matter of fact, some agents may end up forming ``expensive" links and getting the same amount of information that they could have gathered by forming a ``cheaper" link. When information redundancy increases, the value of the information gathered by agents decreases, thus, anarchy costs more and the PoA increases. Contrarily, in the $\mathcal{K}_{M}$ region, the upper bound on PoA decreases as the information redundancy increases in a similar manner to the homogenous link costs scenario. Unlike the PoA, the MIL upper bound is not sensitive to cost heterogeneity since it is only sensitive to informational losses. It can be easily shown that the MIL in recipient-dependent CINs behaves in the same way as in the homogeneous cost scenario. In the next section, we tackle the problem of joint information production and link formation in CINs.

\section{Joint Information Production and Link Formation Games in CINs} 

In the network formation game so far, we have assumed that agents in a CIN are gifted with an exogenously determined entropic vector. Nevertheless, in many practical CINs, agents decide the amount of information to ``produce" given some production cost, e.g. mobile users in cellular systems may download data for social-based services by themselves via the cellular network infrastructure, or get this data opportunistically from other users by establishing D2D links \cite{refnew1}. In this section, we focus on a CIN where each agent jointly decides the amount of information to produce and the links to form.

\subsection{Game formulation}

When agents choose what information to produce, a crucial aspect that affects the network topology and information production is how information aggregates. \cite{ref19} assumes that information aggregates simply by addition; this will be the case only if the value of each additional piece of information is constant; thus, there are no complementarities nor redundancies. \cite{ref14} assumes a specific functional form, the \textit{Dixit-Stiglitz} function; this captures informational complementarities and redundancies in a very special way, i.e. agents appreciate ``diversity of information sources" rather than the ``diversity of the information". In this paper, we consider two modes of aggregation that seem more natural and are suggested by the formulation of information in terms of entropy.

The information production decision taken by $N$ agents in a CIN
corresponds to the selection of a point inside the entropic region
$\Gamma_{N}^{*}$. Correlations between the random variables of different
agents are exogenously determined by external factors, e.g. geographical
locations of sensors. To capture information redundancy, we define
an aggregation function $F_{\mathcal{H}}:\mathbb{R}_{+}^{N}\rightarrow\mathbb{R},$
that maps the entropies of a set of agents to a joint entropy of these
agents, i.e. $H(X_{1},X_{2},...,X_{N})=F_{\mathcal{H}}\left(H(X_{1}),H(X_{2}),...,H(X_{N})\right)$.
Clearly, the range of the function $F_{\mathcal{H}}(.)$ should belong to $\Gamma_{N}^{*}$. Throughout this section, we study two
different aggregation functions: the first is the one corresponding
to independent random variables $H(X_{1},X_{2},...,X_{N})=\sum_{i=1}^{N}H(X_{i})$,
and the second is the one corresponding to strongly correlated random
variables $H(X_{1},X_{2},...,X_{N})=\max\{H(X_{1}),H(X_{2}),...,H(X_{N})\}$.
Both aggregation functions provide insights on how information redundancy
affects the information production decisions at equilibrium. 

In real-world networks, the aggregation function captures the informational
relationships between different agents in a CIN. For instance, in
a sensor network where sensors are deployed over a correlated random
field \cite{ref18}, the information production decision can be thought
of as the precision at which a sensor quantizes its measurements.
Larger precision corresponds to larger value for the entropy. However,
no matter what precision a sensor uses, its measurements will be correlated
to that of another nearby sensor. Thus, the joint entropy of the two
sensors would be governed not only by the precision they decide, but
also by the redundancy in their information that is determined exogenously
by their geographical locations and the nature of the physical process
that they sense. The aggregation function captures such exogenous
factors, and based on it, the behavior of cognitive agents is determined.

In the information production and link formation game, the strategy
of an agent $i$ is denoted by ${\bf s}_{i}=\left(H(X_{i}),{\bf g}_{i}\right)$.
A strategy profile of the game is written as ${\bf s}=(H(X_{1}),H(X_{2}),...,H(X_{N}),{\bf g})$,
and the strategy space is ${\bf S}$. We denote the joint information
production and link formation game by $\bar{\mathcal{G}}^{N}=\langle\mathcal{N},{\bf S},{\bf u}\rangle$.
Thus, different from $\mathcal{G}^{N}$, agents do not observe an
entropic vector, but they decide the entropic vector based on their knowledge
of the aggregation function.

The utility function of agent $i$ is given by 
\begin{equation}
u_{i}({\bf s})=f\left(H(X_{i\cup\mathcal{R}_{i}({\bf g})})\right)-kH(X_{i})-|\mathcal{N}_{i}({\bf g})|c,\label{prodeq}
\end{equation}
where $k$ is the cost of producing one unit of information, $|\mathcal{N}_{i}({\bf g})|$
is the number of agents which agent $i$ form links with, and $H(X_{i\cup\mathcal{R}_{i}({\bf g})})$
is determined by $F_{\mathcal{H}}$ given the production levels of
all agents. We adopt the NE as a solution concept. Thus, a strategy
profile ${\bf s}^{*}$ is an NE profile if no agent benefits from
unilaterally forming a link, breaking a link, or altering the amount
of information it produces. The set of NE profiles is denoted by ${\bf S}^{*}$.
Finally, we denote by $\bar{H}$ the maximum amount of information that
each agent can produce at equilibrium, thus $\bar{H}$
can be obtained by solving $f^{'}(\bar{H})=k$ \cite{ref19}. In the following
subsection, we revisit the motivating example of the two agents interaction
in order to understand the cognitive behavior of agents in $\bar{\mathcal{G}}^{2}$.

\begin{figure}[!t]
\centering
\includegraphics[width=3in]{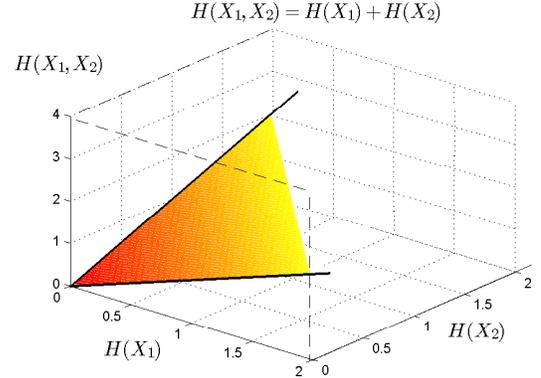}
\captionsetup{font= small}
\caption{The aggregation function for independent random variables.}
\label{fig_pase215}
\end{figure}

\begin{figure}[!t]
\centering
\includegraphics[width=3in]{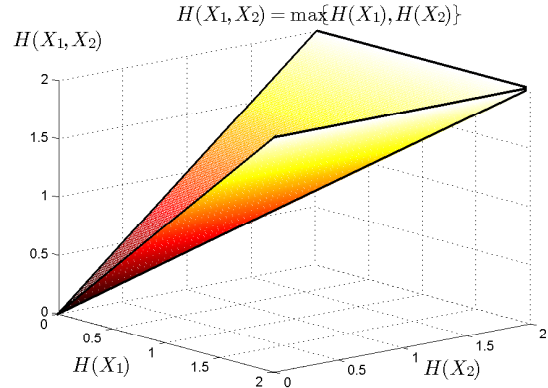}
\captionsetup{font= small}
\caption{The aggregation function for strongly correlated random variables.}
\label{fig_pase315}
\end{figure}

\subsection{Motivating example for two-agents interaction: To produce or not
to produce?}

Consider a simple CIN with only two agents ($N=2$) who are playing
the game $\bar{\mathcal{G}}^{2}$. We aim at characterizing the equilibria
of $\bar{\mathcal{G}}^{2}=\langle\{1,2\},{\bf S},{\bf u}\rangle$,
and investigate the impact of $F_{\mathcal{H}}$, $k$, and $c$ on
the cognitive behavior of the agents. Specifically, we are interested
in identifying scenarios in which one agent may decide not to produce
any information and fully rely on the other. Let us focus on agent
1. The utility function of this agent is given by 
\[
u_{1}({\bf s})=f\left(H(X_{1\cup\mathcal{R}_{1}({\bf g})})\right)-kH(X_{1})-g_{12}c,
\]
where $\mathcal{R}_{1}({\bf g})=\phi$ if $g_{12}=g_{21}=0$, and
$\mathcal{R}_{1}({\bf g})=2$ otherwise. The best response of agent
1 is given by 
\[
u_{1}({\bf s}^{*})=\max_{g_{12},H(X_{1})}\left(f\left(H(X_{1\cup\mathcal{R}_{1}({\bf g})})\right)-kH(X_{1})-g_{12}c\right).
\]
Note that the decision of agent 1 depends on the value of $H(X_{1\cup\mathcal{R}_{1}({\bf g})})$,
which is determined by $F_{\mathcal{H}}$. For 2 agents, the entropic
vector is $\overrightarrow{{\bf H}}=[H(X_{1}),H(X_{2}),H(X_{1},X_{2})]$.
The function $F_{\mathcal{H}}$ maps the information production decisions
$H(X_{1})$ and $H(X_{2})$ to $H(X_{1},X_{2})$. Thus, we have $H(X_{1},X_{2})=F_{\mathcal{H}}\left(H(X_{1}),H(X_{2})\right)$.
In the following, we focus on two different aggregation functions $F_{\mathcal{H}}\left(H(X_{1}),H(X_{2})\right)=H(X_{1})+H(X_{2})$
and $F_{\mathcal{H}}\left(H(X_{1}),H(X_{2})\right)=\max\{H(X_{1}),H(X_{2})\}$. 
\subsubsection{${\bf F_{\mathcal{H}}\left(H(X_{1}),H(X_{2})\right)=H(X_{1})+H(X_{2})}$}
 In this case, the information of agents 1 and 2 are not redundant,
which means that the random variables $X_{1}$ and $X_{2}$ are independent.
Thus, $F_{\mathcal{H}}$ maps the production profile of both agents
to a point in the set $\tilde{\mathcal{H}}$. This reduces to the
aggregation function used in \cite{ref19}. Fig. \ref{fig_pase215}
plots $F_{\mathcal{H}}$, which corresponds to the upper surface of the convex cone $\Gamma_{2}^{*}$ (or equivalently, the hyperplane $\tilde{\mathcal{H}}$).
Assume that the link cost is given by $c>k\bar{H}$. In this case,
we have a unique equilibrium in which $g_{12}^{*}=g_{21}^{*}=0$, 
and $H^{*}(X_{1})=H^{*}(X_{2})=\bar{H}$. Thus, we have a fully disconnected
network with both agents producing information. This means that when
the link cost is very high, every agent decides to produce information
and not to get information from the other. Now assume that $c<k\bar{H}$.
It is easy to show that $g_{12}^{*}g_{21}^{*}=0$, $g_{12}^{*}=1$
or $g_{21}^{*}=1$, and $H^{*}(X_{1})+H^{*}(X_{2})=\bar{H}$. Thus,
when the link cost is low, agents generally produce some of the information
they need and get some other information from the other agent. However,
one possible equilibrium has one agent producing an amount $\bar{H}$
of information with the other forming a link with it and not producing any information on its own.
\subsubsection{${\bf F_{\mathcal{H}}\left(H(X_{1}),H(X_{2})\right)=\max\{H(X_{1}),H(X_{2})\}}$}
Agents may possess fully correlated information in which the joint
entropy is always bounded by the entropy of one of them. Fig. \ref{fig_pase315}
plots $F_{\mathcal{H}}$ which corresponds to the lower surface of
the convex cone $\Gamma_{2}^{*}$. In this case, it is never beneficial
for any agent to form a link and produce a positive amount of information
simultaneously. For $c>k\bar{H}$, we have a unique equilibrium comprising
a fully disconnected network with each agent producing $\bar{H}$.
For $c<k\bar{H}$, we have only one agent producing positive amount
of information in every equilibrium. 

Thus, information redundancy influences the agents' information
production decisions. When the information contains no redundancies, there exist
many equilibria in which both agents produce positive amount of information
when $c<k\bar{H}$. However, for $c<k\bar{H}$, when agents have strongly
correlated information, every equilibrium has only one agent producing
information. Thus, redundancy discourages information sharing between
agents and reduces the number of agents producing information when
the link cost is low. When $c>k\bar{H}$, we always have a disconnected
network with all agents producing information for both aggregation
functions. However, the total amount of information in the network
when the random variables of both agents are independent is $H(X_{1},X_{2})=2\bar{H}$,
while when the information of both agents are fully correlated (i.e.,
$H(X_{1},X_{2})=\max\{H(X_{1}),H(X_{2})\}$), we have $H(X_{1},X_{2})=\bar{H}$.
In the next subsection, we generalize these results to the $\bar{\mathcal{G}}^{N}$
game.

\subsection{Characterization of the NE for $\bar{\mathcal{G}}^{N}$ and asymptotic information production behavior}

In this subsection, we characterize the NE for the $\bar{\mathcal{G}}^{N}$
game. We study the equilibria for the two aggregation functions $F^{1}_{\mathcal{H}}\left(H(X_{1}),H(X_{2}),...,H(X_{K})\right)=\sum_{i=1}^{K}H(X_{i})$,
and $F^{2}_{\mathcal{H}}\left(H(X_{1}),H(X_{2}),...,H(X_{K})\right)=\max\{H(X_{1}),H(X_{2}),...,H(X_{N})\}$.
In the following Theorem, we obtain some properties of the equilibria
of $\bar{\mathcal{G}}^{N}$ when the aggregation function is $F^{1}$. 
\begin{thm}
For the aggregation function $F^{1}$ we have:

(1) If $c>k\bar{H}$, then there exists a unique equilibrium ${\bf s}^{*}$ where the network is fully disconnected  and every agent produces the individually optimal amount of information ($H^{*}(X_{i})=\bar{H)}$. 

(2) If $c<k\bar{H}$, then ${\bf s}^{*}$ is an equilibrium if and only if: (i) the CIN is minimally connected, (ii) the total amount of information is $H (\mathcal{X})=\bar{H}$, and (iii) if any agent $i$ forms a link in the network ($g_{ij}^{*}=1,i,j\in\mathcal{N}$), then the cost of linking should be less than the cost of producing the amount of information obtained by forming a link $c\leq kH^{*}(X_{-i})$. 

\end{thm}
\begin{proof} See Appendix J. \, \IEEEQEDhere \end{proof} Condition
(1) results from indirect information sharing among connected agents.
In addition, the network has a total information of $\bar{H}$ since
all agents perfectly share the information they produce, which results
in condition (2). Finally, condition (3) says that the cost of linking
should be less than the cost of producing the amount of information
obtained via linking. In the following Theorem, we characterize the
equilibrium when the aggregation function is $F^{2}$. 
\begin{thm}
For the aggregation function $F^{2}$ we have:

(1) If $c>k\bar{H}$, then there exists a unique equilibrium ${\bf s}^{*}$ where $g_{ij}^{*}=0,$ and $H^{*}(X_{i})=\bar{H},\forall i,j\in\mathcal{N}$. 

(2) If $c<k\bar{H}$, then ${\bf s}^{*}$ is an equilibrium if and only if: (i) the CIN is minimally connected, (ii) there exists exactly one agent $i$ with $H^{*}(X_{i})=\bar{H}$, and $H^{*}(X_{-i})=0$, (iii) all agents with zero information production form exactly one link. 
\end{thm}
\begin{proof} See Appendix K. \, \IEEEQEDhere \end{proof}

Theorem 7 states that when agents' information is strongly correlated, information production is monopolized by exactly one agent. That is, unlike the case of uncorrelated information, agents do not distribute the production of information among multiple agents who produce complementary information. Thus, we conclude that information redundancy can have significant impact on the information production behavior at equilibrium.

Several questions arise in networks where cognitive agents take joint information production and link formation decisions: what is the fraction of agents producing information at equilibrium in an asymptotically large network? What is the asymptotic total amount of information in the network? In the rest of this subsection, we address these questions and provide a characterization for the asymptotic informational behavior of agents in a CIN. We investigate the asymptotic behavior of two basic quantities: the fraction of agents producing information at equilibrium, and the total amount of information in the network.

Denote the set of agents producing information at equilibrium by $\mathcal{I}({\bf s}^{*})=\left\{ i\,\left|\, i\in\mathcal{N},\,\mbox{and}\,\, H^{*}(X_{i})>0\right.\right\}$. \cite{ref19} show that if agents produce non-redundant information and there is no indirect information sharing, then in equilibrium, information is produced by only a small subset of agents, and the fraction of information producers becomes vanishingly small as the network size grows, i.e. $\lim_{N\rightarrow\infty}\sup_{{\bf s}^{*}\in{\bf S}^{*}}\frac{|\mathcal{I}({\bf s}^{*})|}{N}=0$. \cite{ref19} calls this ``\textit{the law of the few}".
In the next corollary, we characterize the fraction of information producers and the total amount of information in the network in $\bar{\mathcal{G}}^{\infty}$ when the link cost is large.

\begin{crlry} In the $\bar{\mathcal{G}}^{N}$ game, when $c>k\bar{H}$,
we have 
\begin{equation}
\lim_{N\rightarrow\infty}\frac{|\mathcal{I}({\bf s}^{*})|}{N}=1,\label{thermeq1}
\end{equation}
for both $F^{1}_{\mathcal{H}}$ and $F^{2}_{\mathcal{H}}$.
For $F^{2}_{\mathcal{H}}$, the total amount of information in
the network in $\bar{\mathcal{G}}^{\infty}$ is given by 
\begin{equation}
\lim_{N\rightarrow\infty}H(X_{1},X_{2},...,X_{N})=\bar{H},\label{thermeq3}
\end{equation}
while for $F^{1}_{\mathcal{H}}$ we have 
\begin{equation}
\lim_{N\rightarrow\infty}H(X_{1},X_{2},...,X_{N})=\infty.\label{thermeq5}
\end{equation}
\end{crlry}

\begin{proof} See Appendix L. \, \IEEEQEDhere \end{proof}

Corollary 2 says that when the link cost is very high, the network
is fully disconnected and every agent produces the information it
needs. Thus, when the network is asymptotically large, every agent
is an information producer no matter what the amount of information
redundancy is. The number of agents producing information is always
$N$. While the number of information producers does not depend on
$F_{\mathcal{H}}$, it is clear that the total amount of information
in the network depends on the amount of redundancy. When the agents'
information are strongly correlated, the total amount of information
is always bounded by $\bar{H}$. On the other hand, when agents have
uncorrelated information, the total amount of information in an asymptotically large network
is unbounded. In the next corollary, we study the case the information production behavior when the link cost is low.

\begin{crlry} In the $\bar{\mathcal{G}}^{N}$ game, when $c<k\bar{H}$,
the fraction of information producers for $F^{2}_{\mathcal{H}}$
is given by 
\begin{equation}
\lim_{N\rightarrow\infty}\sup_{{\bf s}^{*}\in{\bf S}^{*}}\frac{|\mathcal{I}({\bf s}^{*})|}{N}=0,\label{thermeq12}
\end{equation}
while for $F^{1}_{\mathcal{H}}$ we have 
\begin{equation}
\lim_{N\rightarrow\infty}\sup_{{\bf s}^{*}\in{\bf S}^{*}}\frac{|\mathcal{I}({\bf s}^{*})|}{N}=1.\label{thermeq22}
\end{equation}
For both $F^{2}_{\mathcal{H}}$ and $F^{1}_{\mathcal{H}}$,
the total information in the network is 
\begin{equation}
\lim_{N\rightarrow\infty}H(X_{1},X_{2},...,X_{N})=\bar{H}.\label{thermeq32}
\end{equation}
\end{crlry} \begin{proof} 
See Appendix M.
 \, \IEEEQEDhere \end{proof}

This corollary states that the law of the few introduced in \cite{ref19} does not generally apply in the case of indirect information sharing. The applicability of the law of the few depends on the link cost and information redundancy. For instance, in a network with agents producing highly redundant information, the law of the few only applies when the link cost is $c < k\bar{H}$, whereas for $c > k\bar{H}$, all agents will be information producers as shown in Fig. \ref{figexemp2}, where we display the network topology at equilibrium with each agen labeled by the amount of information it produces for $c > k\bar{H}$ and $c < k\bar{H}$.       
Moreover, there exists information aggregation functions in which the network at NE can have all the agents being information producers for any link cost, and production is no longer dominated by a small set of \textit{hub} agents. Thus, even for low link costs, the applicability of the law of the few is still governed by the amount of information redundancy. If the agents' information are strongly correlated, the law of the few applies and information production is dominated by a small fraction of agents in every equilibrium for $c < k\bar{H}$. In contrast, when the agents produce non-redundant information, the law of the few fails even for low link costs, i.e. $c < k\bar{H}$. Fig \ref{figexemp1} depicts the equilibria for an 8-agent network with $c < k\bar{H}$ when the aggregation function is $F^{2}_{\mathcal{H}}$ and $F^{1}_{\mathcal{H}}$. It is observed that the law of the few applies when the aggregation function is $F^{2}_{\mathcal{H}}$, but fails when the aggregation function is $F^{1}_{\mathcal{H}}$.   
 
Note that while we focused on the extreme cases of information redundancy by considering the aggregation functions $F^{1}_{\mathcal{H}}$ and $F^{2}_{\mathcal{H}}$, the analysis can be extended to other generic aggregation functions. Such generic aggregation functions should be derived from a real-world network setting (e.g. geographical deployment of sensor networks), and an interesting problem becomes studying the information production behavior of agents under these aggregation functions. However, it is sufficient to only consider $F^{1}_{\mathcal{H}}$ and $F^{2}_{\mathcal{H}}$ to show that the celebrated law of the few does not generally hold whenever information redundancy is considered.   



\begin{figure}[!t]
\centering
\includegraphics[width=3.5in]{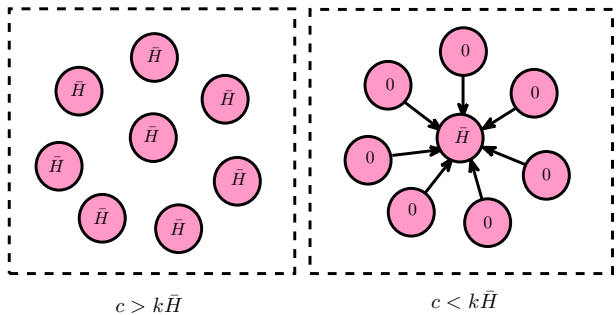}
\captionsetup{font= small}
\caption{Connectivity and information production behavior in a network with strongly correlated information sources.}
\label{figexemp2}
\end{figure}

\begin{figure}[!t]
\centering
\includegraphics[width=3.5in]{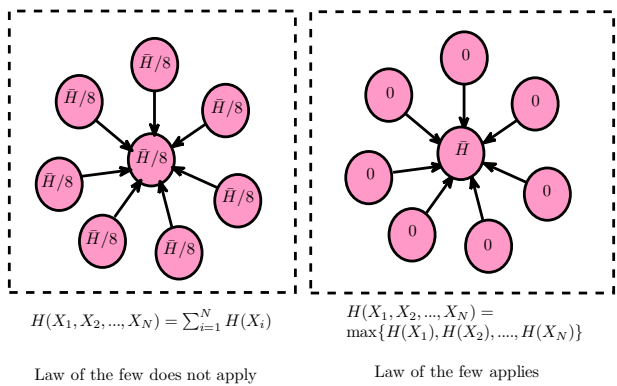}
\captionsetup{font= small}
\caption{Exemplary equilibria for different aggregation functions. The law of the few does not apply when the agents' information has no redundancies.}
\label{figexemp1}
\end{figure}

\section{Future work and extensions}
In this section, we propose some potential future research directions that capitalize on our model. \\ 
{\bf 1- Dynamic games with incomplete information:} we have considered a one-shot complete information game in which all agents have knowledge of the entropic vector. An extension to our model is to consider a dynamic game with incomplete information \cite{ref16}, in which agents {\it learn} the entropic vector over time by interacting with other agents. In this case, agents would pay a link {\it maintenance} cost to keep connected to informative agents, and would break links with non-informative ones. In such model, the network can be characterized in terms of the probability of emergence of certain topology at NE, and the time needed for the network to converge to a steady-state topology. \\    
{\bf 2- Incorporating capacity-constrained links:} we have assumed perfect indirect information sharing among agents. In some settings, such as multi-hop relaying networks, information sharing can be lossy and the links between agents can be capacity-constrained. While lossy benefit flow has been modeled before by assuming that benefits are discounted at each links \cite{ref12}, in our model lossy information sharing can be modeled using an information-theoretic approach by treating links as erroneous channels. Incorporating these factors into our model can lead to interesting results on both the network topology at NE and information production behavior. \\
{\bf 3- New solution concepts:} the network formation game considered in this paper adopts the NE as a solution concept. However, different networks and applications can be better suited by different solution concepts. For instance, in many applications, such as D2D communications, establishing a link requires a mutual consent among agents. In this case, {\it pairwise stability} can be used as a solution concept instead of the NE. 

\section{Conclusions}

In this work, we present a first model for the endogenous formation
of networks by cognitive agents who aim at gathering and producing
information. Using Nash Equilibrium as a solution concept, we
formulated a non-cooperative network formation game where agents get
informational benefits by forming costly links with each other. We
show that the information possessed by the cognitive agents affects
the network topology, efficiency, and information production behavior.
We show the impact of information redundancy on the topologies of
NE networks, and its impact on the network efficiency in terms of
the Price-of-Anarchy and Maximum Information Loss. Finally,
we consider the asymptotic behavior of a network where each agent
both produces information and forms links with other agents. For such
networks, we study the impact of information redundancy on the number
of agents producing information at equilibrium. We show that the validity of the
law of few depends on how information aggregates.

\appendices{}

\section{Proof of Theorem 1}

\global\long\def\theequation{\thesection.\arabic{equation}}
 \setcounter{mytempeqncnt}{\value{equation}} \setcounter{equation}{0}
From {\it Nash's Existence Theorem}, we know that if we allow mixed strategies, then every game with a
finite number of players in which each player can choose from finitely
many pure strategies has at least one Nash equilibrium \cite{ref22X}. Assume that
agent $i$ adopts a mixed strategy $\Delta_{i}=(p_{i1},p_{i2},...,p_{iN})$,
where $p_{ij}$ is the probability that agent $i$ forms a link with
agent $j$, and $p_{ii} = 0, \forall i \in \mathcal{N}$. The utility of agent $i$ in this case is obtained by
averaging over all possible networks as follows 
\begin{equation}
u_{i}(\Delta_{i})=\sum_{j=1}^{2^{N-1}-1}w_{j}f\left(H(X_{i}\cup X_{\alpha_{j}})\right)-\sum_{l=1}^{N-1}p_{il}c,\label{eqap1}
\end{equation}
where $\alpha_{j}$ is an element of the power set of $\mathcal{N}/\{i\}$,
and $w_{j}$ is the probability of the emergence of a network component comprising agents in the set $\left\{i \cup \alpha_{j}\right\}$
based on the mixed strategies. For instance, in a 2 agent network,
the utility function of agent 1 is given by 
\[u_{1}(\Delta_{1})=\]
\[\left(p_{12}(1-p_{21})+p_{21}(1-p_{12})+p_{12}p_{21}\right)f\left(H(X_{1},X_{2})\right)
\]
\[
+(1-p_{12})(1-p_{21})f\left(H(X_{1})\right)-p_{12}c.
\]
In this case, $w_{1}=p_{12}(1-p_{21})+p_{21}(1-p_{12})+p_{12}p_{21}$ and
$w_{2}=(1-p_{12})(1-p_{21})$. Let the NE strategy profile be ${\bf \Delta}^{*}=(\Delta^{*}_{1}, \Delta^{*}_{2},...,\Delta^{*}_{N})$, where $\Delta^{*}_{i}=(p^{*}_{i1},p^{*}_{i2},...,p^{*}_{iN}).$ According to (\ref{eq6}), the following condition on ${\bf \Delta}^{*}$ needs to be satisfied   
\begin{equation}
u_{i}(\Delta^{*}_{i}, \Delta^{*}_{-i}) \geq u_{i}(\Delta_{i}, \Delta^{*}_{-i}), \forall \Delta_{i} \in [0,1]^{N}, \forall i \in \mathcal{N}.
\label{nashcond}
\end{equation}
Now we show that for any agent $i$, the NE strategy $\Delta^{*}_{i}$ needs to be a pure strategy for condition (\ref{nashcond}) to be satisfied. We focus on agent $i$ with a NE strategy $\Delta^{*}_{i}=(p^{*}_{i1},p^{*}_{i2},...,p^{*}_{iN}),$ where $p^{*}_{ij} \in [0,1]$. Now assume we induce a perturbation $\epsilon$ to the mixed strategy of agent $i$ by modifying $p^{*}_{ik}$ to $p^{*}_{ik}+\epsilon$ for a certain $k$, where $\epsilon\in[-p^{*}_{ik},1-p^{*}_{ik}]$. We call this modified strategy $\Delta^{*}_{i}(\epsilon)$. Note that we can write any $w_{j}$ in (\ref{eqap1}) in the form of $w_{j}=\tilde{w}_{j}p^{*}_{ik}+\bar{w}_{j}(1-p^{*}_{ik})$.
This results in a perturbed utility $u_{i}(\Delta^{*}_{i}(\epsilon))$ as follows 
\[
u_{i}(\Delta^{*}_{i}(\epsilon))=\sum_{j=1}^{2^{N-1}-1}\left(\tilde{w}_{j}(p^{*}_{ik}+\epsilon)f\left(H(X_{i}\cup X_{\alpha_{j}})\right)\right.
\]
\begin{equation}
\left.+\bar{w}_{j}(1-\epsilon-p^{*}_{ik})f\left(H(X_{i}\cup X_{\alpha_{j}})\right)\right)-(p^{*}_{ik}+\epsilon)c-\sum_{l=1,l\neq k}^{N-1}p^{*}_{il}c,\label{eqap2}
\end{equation}
which can be rearranged as 
\[
u_{i}(\Delta^{*}_{i}(\epsilon))=
\]
\[
u_{i}(\Delta^{*}_{i})+\epsilon\left(\sum_{j=1}^{2^{N-1}-1}(\tilde{w}_{j}-\bar{w}_{j})f\left(H(X_{i}\cup X_{\alpha_{j}})\right)-c\right).
\]
Let $\delta=\sum_{j=1}^{2^{N-1}-1}(\tilde{w}_{j}-\bar{w}_{j})f\left(H(X_{i}\cup X_{\alpha_{j}})\right)-c$.
It can be easily shown that $\frac{\partial u_{i}(\Delta^{*}_{i}(\epsilon))}{\partial\epsilon}>0$
if $\delta>0$, and $\frac{\partial u_{i}(\Delta_{i}^{*}(\epsilon))}{\partial\epsilon}<0$
otherwise. Thus, if $\delta>0$, agent $i$ can always increase its
utility by increasing $\epsilon$ and setting $\epsilon=1-p^{*}_{ik}$ (and thus playing a pure strategy with $p_{ik}=1$),
and if $\delta<0$, agent $i$ can always increase its utility by setting $\epsilon = -p^{*}_{ik}$(and thus playing a pure strategy with $p_{ik}=0$), which contradicts with $\Delta^{*}_{i}$ being a NE strategy. Thus, for all $k\in\mathcal{N}/\{i\}$, agent
$i$ needs to select a pure strategy $p^{*}_{ik}\in\{0,1\}$ for $\Delta^{*}_{i}$ to be a best response to $\Delta^{*}_{-i}$ regardless of the strategies of other agents, i.e. non-pure strategies are always dominated by a pure strategy. Due to symmetry, this applies to all agents in $\mathcal{N}$. Therefore, it follows that a pure strategy NE always exists.

\section{Proof of Proposition 1}

\global\long\def\theequation{\thesection.\arabic{equation}}
 \setcounter{mytempeqncnt}{\value{equation}} \setcounter{equation}{0}
If the component $\mathcal{C}$ is not minimally connected, then it
has at least one cycle as there exist agents $i$ and $j$ that are
connected via (at least) two paths $p_{ij,1}$ and $p_{ij,2}$, such that any
of the two paths is not a subset of the other. For such component
at NE, assume that agent $v$ is on path $p_{ij,1}$ and agent $w$
is on path $p_{ij,2}$. Note that all the agents receive the same
amount of total information $H(\mathcal{C})$, and we know that there
indeed exists links: $g_{xv}^{*}$ (or $g_{vx}^{*}$) and $g_{wy}^{*}$
(or $g_{yw}^{*}$), where agent $x\in p_{ij,1}$ and agent $y\in p_{ij,2}$.
Now focus on any link of them, say $g_{wy}^{*}=1$. We observe that
agent $w$ can break this link and still receive the same benefit
by gathering the same amount of information from path $p_{ij,1}$,
thus receiving a strictly higher utility function as it will not pay
the cost for the link with agent $y$, which contradicts with the fact
that ${\bf g}^{*}$ is an NE. Thus, a single path exists between any
two agents.

\section{Proof of Lemma 1}

\global\long\def\theequation{\thesection.\arabic{equation}}
 \setcounter{mytempeqncnt}{\value{equation}} \setcounter{equation}{0}

If there exists an agent in which other agents have an incentive to connect to even if they possess all other information in the network, then the network is indeed connected at any equilibrium. This is satisfied if and only if the linking cost satisfies $c \leq f(H(\mathcal{X}))-f(H(X_{-i}))$ for some agent $i$ in $\mathcal{N}$, i.e. the marginal benefit from connecting to that agent is always more than the link cost irrespective to the current connections of the agent forming the link. Thus, we must have $c<\max_{i}f(H(\mathcal{X}))-f(H(X_{-i}))$. Hence, part (i) of the Lemma follows.

If no agent have an incentive to form any link, then the network is fully disconnected. From the monotonicity property of the entropy, we know that if agent $i$ has no incentive to connect to a set $\mathcal{V}$ of agents, then it has no incentive to connect to a set $\mathcal{U}$ if $\mathcal{U}\subseteq\mathcal{V}$. Thus, if agent $i$ has no incentive to connect to the set $\mathcal{N}/\{i\}$ via a single link, then it has no incentive to form any link in the network. This occurs if and only if $c \geq f(H(\mathcal{X}))-f(H(X_{i}))$. If this condition is satisfied for all agents, then the network is indeed disconnected, and part (ii) of the Lemma follows.

\section{Proof of Theorem 2}

\global\long\def\theequation{\thesection.\arabic{equation}}
For the network to be in NE, no agent should have an incentive to unilaterally deviate by forming a new link or breaking an existing link. We focus on an arbitrary network component at NE, say component $\mathcal{C}_{i}$. Inside this component, each agent should either have an incentive to form at least one link, or other agents should have an incentive to connect to it. Otherwise, this agent can get disconnected (by having this agent breaking a link or other agents break their links with that agent) from the component while strictly increasing its utility or the utility of other agents in the component. Thus, we must have either $f(H(X_{\mathcal{C}_{i}}))-f(H(X_{j})) \geq c$ or $f(H(X_{\mathcal{C}_{i}}))-f(H(X_{\mathcal{C}_{i}/\{j\}})) \geq c$ for all agents $j$ in $\mathcal{C}_{i}$. This should apply to all components in the network. Hence, condition (1) follows.

Now focus on the interaction between different components of the network. If any agent in component $\mathcal{C}_{i}$ benefits from forming a link to any agent in component $\mathcal{C}_{j}$, then the network is not NE since in this case an agent in $\mathcal{C}_{i}$ can strictly increase its utility by unilateral deviation. Hence, we should have $f(H(X_{\mathcal{C}_{i}\cup\mathcal{C}_{j}}))-f(H(X_{\mathcal{C}_{i}}))\leq c$ for any two components in the network. Thus, condition (2) follows.

\section{Proof of Theorem 3}

\global\long\def\theequation{\thesection.\arabic{equation}}
When the three conditions in Theorem 3 are satisfied, then the network is a strict NE since the action of each agent is strictly better than any other action, i.e. core agents in each component strictly better off when sponsoring the periphery agents, and all the periphery agents strictly better off when they do not form any links. Thus, no agent is indifferent to multiple actions, which implies that the NE is strict. Now we prove the converse by showing that if the network is a strict NE, the 3 conditions in Theorem 3 must be satisfied. Under strict NE, a non-singleton component $\mathcal{C}$ has two agents $i$ and $j$ such that $g_{ij}^{*}=1$. Now assume that $g_{kj}^{*}=1$ for some agent $k\in\mathcal{C}$. It is clear that $k$ can achieve the same utility by deleting its link with $j$ and connecting to $i$. This contradicts with the fact that $g^{*}$ is a strict NE. Thus, $g_{kj}^{*}=0$. Using a similar argument, it can be shown that $g_{ki}^{*}=0$. Thus, we conclude that $g_{ik}^{*}=1$. This is true for all agents $k\in\mathcal{C}$, which implies that a single core agent $i$ forms links with all other agents in $\mathcal{C}$. Therefore, the core agent $i$ should strictly increase its utility for each of the $M-1$ links it forms. The marginal utility of agent $i$ from forming a link with agent $j$ given that $i$ is connected to all other agents in $\mathcal{C}$ is given by $f(H(X_{\mathcal{C}}))-f(H(X_{\mathcal{C}/\{j\}}))-c$. This should be positive for all agents $j$ in $\zeta$, where $\zeta$ is the set of the $M-1$ periphery agents, because otherwise core agent $i$ can break some of the links in the component. Thus, for agent $i\in\mathcal{C}$ to be a core agent, and for agents $j\in\mathcal{C}$ to reside in the periphery, we must have $f(H(X_{\mathcal{C}}))-f(H(X_{\mathcal{C}/\{j\}}))>c,\,\forall\, j\neq i$. Note that conditions (1) and (2) in Theorem 2 should also be satisfied for the network to be an NE, while the feasibility of organizing each component as a core-sponsored star guarantees that the network is at strict NE. Thus, strict NE exists if there exists an NE with the set $\zeta$ having a cardinality that is not less than $M-1$ for all components, i.e. a single core agent can sponsor each component.

\section{Proof of Lemma 2}
\global\long\def\theequation{\thesection.\arabic{equation}}
We know that in the $\mathcal{K}_{C}$ region, all the NE networks are minimally connected. For a minimally connected, each agent has an aggregate benefit of $f(H(\mathcal{X}))$ and the total number of links is $N-1$ (total cost is $(N-1)c$), thus the social welfare of any minimally connected network with strategy profile ${\bf g}$ is given by
\[U({\bf g}) = Nf(H(\mathcal{X}))-(N-1)c.\] 
In the following we show that this is indeed the maximum social welfare in the $\mathcal{K}_{C}$ region, which means that the socially optimal network in this region is minimally connected. Note that the maximum sum benefits for all agents in the network is $Nf(H(\mathcal{X}))$, i.e. all agents share all information, thus any connected network maximizes the sum benefit. Recall that in the $\mathcal{K}_{C}$ region, we have $c \leq f\left(H(\mathcal{X})\right)-f(\min_{i}H(X_{-i}))$. Thus, for any (disconnected) network with less than $N-1$ links, the social welfare can always be increased by adding a set of links that makes the network (minimally) connected. On the other hand, we know from the pigeonhole principle that any network with more than $N-1$ has cycles, thus the social welfare can always be increased by breaking a set of links such that all cycles are eliminated while keeping the network minimally connected. Therefore, we conclude that the social optimal network in the $\mathcal{K}_{C}$ region is minimally connected, and $\tilde{U} = Nf(H(\mathcal{X}))-(N-1)c$. Since the social welfare of any NE network in $\mathcal{K}_{C}$ is given by $U({\bf g}^{*})= Nf(H(\mathcal{X}))-(N-1)c$, then every NE network is socially optimal and we have $\mbox{PoA} = 1$.
Next, we focus on the $\mathcal{K}_{I}$ region. In this region, any connection will result a negative payoff for any agent who forms a link since $c>f(H(\mathcal{X}))-f(\min_{i}H(X_{i}))$. Thus, the social optimal is a fully disconnected network, which is also the unique (strict) NE, and the PoA = 1 in the $\mathcal{K}_{I}$ region. For the $\mathcal{K}_{M}$ region, we compute an upper bound on the PoA. The lowest social welfare of any equilibrium network in the $\mathcal{K}_{I}$ region is lower bounded by $\sum_{i=1}^{N}f(H(X_{i}))$, i.e. $\inf_{{\bf g}^{*} \in {\bf G}^{*}} U({\bf g}^{*}) \geq \sum_{i=1}^{N}f(H(X_{i}))$, with equality when $f(H(X_{i},X_{j}))-f(H(X_{i}))<c,\forall i,j \in \mathcal{N}$, and $f(H(\mathcal{X}))-f(H(X_{i}))>c,\forall i$ (i.e., agents do not get immediate benefit from forming links to individual agents, thus a fully disconnected network is an NE since not forming a link is a best response for all agents in a fully disconnected network). On the other hand, the social welfare of the socially optimal network in the $\mathcal{K}_{M}$ region is upper bounded by $Nf(H(\mathcal{X}))$, i.e. the social welfare is always strictly less than the sum benefit of all agents when they possess all the information in the network. Thus, it follows that $\mbox{PoA} < \frac{Nf(H(\mathcal{X}))}{\sum_{i=1}^{N}f(H(X_{i}))}$.

\section{Proof of Corollary 1}
\global\long\def\theequation{\thesection.\arabic{equation}}
In the $\mathcal{K}_{C}$ region, we know that all NE networks are connected. Thus, $\sup_{{\bf g}_{u}^{*}\in G^{*}}H(X_{i}\cup X_{\mathcal{R}_{i}({\bf g}_{u}^{*})})=\inf_{{\bf g}_{u}^{*}\in G^{*}}H(X_{i}\cup X_{\mathcal{R}_{i}({\bf g}_{u}^{*})})=H(\mathcal{X})$, and $\mbox{MIL}=0$. Similarly, in the $\mathcal{K}_{I}$ region, we have $\sup_{{\bf g}_{u}^{*}\in G^{*}}H(X_{i}\cup X_{\mathcal{R}_{i}({\bf g}_{u}^{*})})=\inf_{{\bf g}_{u}^{*}\in G^{*}}H(X_{i}\cup X_{\mathcal{R}_{i}({\bf g}_{u}^{*})})=\min_{i}H(X_{i})$, thus $\mbox{MIL}=0$. In the $\mathcal{K}_{M}$ region, the MIL is maximized if both a connected and a fully disconnected network are equilibria. In this case, $\sup_{{\bf g}_{u}^{*}\in G^{*}}H(X_{i}\cup X_{\mathcal{R}_{i}({\bf g}_{u}^{*})})=H(\mathcal{X})$, and $\inf_{{\bf g}_{u}^{*}\in G^{*}}H(X_{i}\cup X_{\mathcal{R}_{i}({\bf g}_{u}^{*})})=\min_{i}H(X_{i})$. Thus, $\mbox{MIL}\leq H(\mathcal{X})-\min_{i}H(X_{i})$, with equality when $c > f(H(X_{i},X_{j}))-f(H(X_{i})), \forall i,j \in \mathcal{N},$ and $c < f(H(\mathcal{X}))-f(H(X_{i})), \forall i \in \mathcal{N}.$ 

\section{Proof of Lemma 3}
\global\long\def\theequation{\thesection.\arabic{equation}}
For a connected network in the $\mathcal{K}_{C}$ region, the utility
of agent $i$ is given by $u_{i}({\bf g}^{*})=f(H(\mathcal{X}))-\sum_{m\in\mathcal{N}_{i}({\bf g}^{*})}c_{m}$.
The social welfare is given by $U({\bf g}^{*})=\sum_{i\in\mathcal{N}}u_{i}({\bf g}^{*})$.
Since we know from Proposition 1 that the network is minimally connected
at equilibrium, then it has exactly $N-1$ links. Therefore, we have
$U({\bf g}^{*})=Nf(H(\mathcal{X}))-\sum_{j\in\mathcal{J}}c_{j}$,
where $\mathcal{J}$ is the set of links in the network designated
by the index of link recipient, and $|\mathcal{J}|=N-1$. The social optimal topology is the connected network with minimum total link costs, which corresponds to a periphery-sponsored star with the agent $k=\arg\min_{j}c_{j}$ residing in the core of the star. The social welfare of such topology is $U({\bf g}^{*})=Nf(H(\mathcal{X}))-(N-1)\min_{j}c_{j}$.
Note that this is also an NE equilibrium as each agent does not benefit
from breaking its link with the core agent and linking to any other
periphery agents. Next, we identify the equilibrium with the worst social welfare. Assume that the link costs are arranged
ascendingly as $c_{1}<c_{2}<c_{3}<...<c_{N-1}<c_{N}$. We know that
the network is minimally connected, thus total costs of link formation
is given by $\sum_{j\in\mathcal{J}}c_{j}$. What are the elements
of the set $\mathcal{J}$ such that the total cost is maximized and
the network is at equilibrium? Note that for the socially optimal
profile, $\mathcal{J}=\{1,1,1,...,1\}$, with a cardinality of $N-1$.
Now assume a line network with $g_{i,i+1}^{*}=1,\forall1\leq i<N$.
Thus, we have $g_{12}^{*}=g_{23}^{*}=...=g_{N-1,N}^{*}=1$. Thus,
$\mathcal{J}=\{2,3,...,N-1,N\}$. It can be easily shown that this
line network is stable, since no agent $i$ can break its link with
agent $i+1$ and increase its utility. For instance, if $i$ breaks
its link with $i+1$, it must connect to any agent $j>i+1$ to receive
the same amount of information but at a higher cost. It can be also
shown that this is the worst equilibrium. This is because for a connected
network, only one agent $i$ connects to the agent $N$ with the highest
link cost, and others can connect to $i$ (which has a lower link cost) and get the information
of $N$ via indirect sharing. The same applies to agent $N-1$, where
one agent connects to it, and others share information by connecting
to that agent. Thus, to maximize the total link cost and maintain
equilibrium, only one link is formed with each agent except the one
with the minimum link cost. Thus, the social welfare in this case
is $NH(\mathcal{X})-\sum_{j=1}^{N}c_{j}+\min_{k}c_{k}$, and the PoA
formula follows.

For the $\mathcal{K}_{I}$ and $\mathcal{K}_{M}$ regions, the proof is the same as that of Lemma 2.

\section{Proof of Theorem 5}

\global\long\def\theequation{\thesection.\arabic{equation}}
In the $\mathcal{K}_{C}$ region, the PoA can be written as $\mbox{PoA}=\frac{Nf(\sum_{i=1}^{N}H(X_{i})-\mbox{KL}(\mathcal{X}))-(N-1)\min_{k}c_{k}}{Nf(\sum_{i=1}^{N}H(X_{i})-\mbox{KL}(\mathcal{X}))-\sum_{j=1}^{N}c_{j}+\min_{k}c_{k}}$. It can be easily shown that if the KL divergence varies from $\mbox{KL}(\mathcal{X}) = \mbox{KL}_{1}$ to $\mbox{KL}(\mathcal{X}) = \mbox{KL}_{2}$, where $\mbox{KL}_{1} < \mbox{KL}_{2}$ and the values of the individual agents' entropies are fixed, then the PoA increases, i.e. we have $\frac{Nf(\sum_{i=1}^{N}H(X_{i})-\mbox{KL}_{1}(\mathcal{X}))-(N-1)\min_{k}c_{k}}{Nf(\sum_{i=1}^{N}H(X_{i})-\mbox{KL}_{1})-\sum_{j=1}^{N}c_{j}+\min_{k}c_{k}} < \frac{Nf(\sum_{i=1}^{N}H(X_{i})-\mbox{KL}_{2})-(N-1)\min_{k}c_{k}}{Nf(\sum_{i=1}^{N}H(X_{i})-\mbox{KL}_{2})-\sum_{j=1}^{N}c_{j}+\min_{k}c_{k}}$.

\section{Proof of Theorem 6}

\global\long\def\theequation{\thesection.\arabic{equation}}
 We start with the case of $c>k\bar{H}$. Assume that there exists
a link in ${\bf g}^{*}$ with $g_{ij}^{*}=1$. In this case, agent
$i$ can always better off by breaking this link and producing an
amount $\bar{H}$ of information. This applies to any agent $i$ in
$\mathcal{N}$. Thus, we have a unique equilibrium with $g_{ij}^{*}=0$,
and $H^{*}(X_{i})=\bar{H}$, $\forall i,j\in\mathcal{N}$. 

Now focus on the case of $c<k\bar{H}$. We show that if ${\bf s}$ satisfies
(i), (ii), and (iii), then ${\bf s}$ is an NE. The minimality of each
network component can be easily proved using Proposition 1. Now we
show that the connected network is an NE. In a disconnected network,
an agent has to produce an amount $\bar{H}$ of information, which
is not optimal since $c<k\bar{H}$. Thus, no agent in a connected
network has incentive to break its link and part (i) follows. Since
the network is minimally connected, then each agent obtains all the
total amount of information $H(\mathcal{X})$. If $H(\mathcal{X})=\bar{H}$,
then no agent in the component has incentive to alter their information
production profile because all agents benefit only from obtaining
an amount $\bar{H}$ of information. Thus, part (ii) is proved. Finally,
if $c\leq kH^{*}(X_{-i}),\forall i$, then no agent in the network
has incentive to break the link it forms and produce an amount $H^{*}(X_{-i})$
of information on its own. Thus, ${\bf s}$ is a Nash equilibrium.

We now prove the converse. Let ${\bf s}$ be an equilibrium. Assume
that the network has two components $\mathcal{C}_{1}$ and $\mathcal{C}_{2}$.
The total amount of information in each component must
be $\bar{H}$ at equilibrium, thus, any agent with positive amount
of information production in one component will better off by not
producing any information and forming a link to the other component.
Thus, the network is connected in NE and part (i) follows. Due to
indirect information sharing, part (ii) is directly concluded. Finally,
if ${\bf s}$ is an equilibrium and $g_{ij}=1$, then this should
be optimal for agent $i$, thus $c\leq kH^{*}(X_{-i}),\forall i\mathcal{N}$. 

\section{Proof of Theorem 7}
\global\long\def\theequation{\thesection.\arabic{equation}}
 The case of $c>k\bar{H}$ is exactly the same as in Theorem 6 and
the proof will be similar to that in Appendix J. Now focus on the
case of $c<k\bar{H}$. We show that if ${\bf s}$ satisfies (i), (ii),
and (iii), then ${\bf s}$ is an NE. Part (i) follows from Proposition
1 and the proof of Theorem 6. Now assume that only one agent in the
network produces $\bar{H}$ information and all others do not produce
any information and only form links in the network. In this case,
the agent producing information does not better off by producing any
amount of information other than $\bar{H}$. In addition, the agents
forming links do not better off by forming new links or breaking their
links and producing information since $c<k\bar{H}$. Thus, part (ii) follows. Since there are
$N-1$ agents forming links, then the network is connected, and no
agent benefits from forming an extra link in the network, which concludes
part (iii).

We now prove the converse. Let ${\bf s}$ be an equilibrium. Due to
indirect information sharing, part (i) follow straightforwardly. Assume
that we have two agents with $H^{*}(X_{i})>H^{*}(X_{j})>0$, then
agent $j$ can always better off by setting $H^{*}(X_{j})=0$ since
the aggregate information of $i$ and $j$ is $H^{*}(X_{j})$. Therefore,
the agent with maximum information production has to set $H^{*}(X_{i})=\bar{H}$,
and all others do not produce information and form a link in the network
since $c<k\bar{H}$. Finally, since ${\bf s}$ is an equilibrium,
agents act optimally (their actions are best responses to the actions of others), thus each agent from the set of $N-1$ non-producers
forms exactly one link in the network.

\section{Proof of Corollary 2}
\global\long\def\theequation{\thesection.\arabic{equation}}
From Theorem 6, we know that when $c>k\bar{H}$, then
we have a unique equilibrium ${\bf s}^{*}$ for both $F^{1}_{\mathcal{H}}$
and $F^{2}_{\mathcal{H}}$ in which $g_{ij}^{*}=0,\,\forall i,j\in\mathcal{N},$
and $H^{*}(X_{i})=\bar{H}$. Thus, we have $H^{*}(X_{i})>0,\forall i\in\mathcal{N},$
and $\frac{|\mathcal{I}({\bf s}^{*})|}{N}=1$, which applies when
the number of agents in the CIN grows to infinity, hence (\ref{thermeq1})
follows. Next, we focus on the total amount of information in the
network. For $F^{2}_{\mathcal{H}}$, we have $H(X_{1},X_{2},...,X_{N})=\max\{\bar{H},\bar{H},...,\bar{H}\}=\bar{H}$,
and (\ref{thermeq3}) follows. Finally, for $F^{1}_{\mathcal{H}}$,
we have $H(X_{1},X_{2},...,X_{N})=\sum_{i=1}^{N}\bar{H}=N\bar{H}$,
and (\ref{thermeq5}) follows.

\section{Proof of Corollary 3}
\global\long\def\theequation{\thesection.\arabic{equation}}
We start by deriving (\ref{thermeq12}).
From Theorem 7, we know that for $F^{1}_{\mathcal{H}}$, every
equilibrium has only one information producer. When the number of
agents grows to infinity, we will still have one information producer
and $\frac{|\mathcal{I}({\bf s}^{*})|}{N}=0$. In order to prove (\ref{thermeq22}),
one needs to find one network in equilibrium for $F^{1}_{\mathcal{H}}$
in which, for arbitrary $N$, we have $N$ information producers.
Consider this network for $N$ agents. Assume that $H(X_{i})=\frac{\bar{H}}{N},\forall i\in\mathcal{N}$,
and the network has a single component which is periphery-sponsored
star network. For this network, we have $|\mathcal{I}({\bf s})|=N$.
We want to show that this network is an NE by showing that every agents
strategy is best response to all others. It is easy to see that since
$c<k\bar{H}$, each periphery agent has no incentive to break its
link with the core since $\frac{N-1}{N}k\bar{H}>c$ when $N$ is asymptotically
large. Moreover, no agent has incentive to alter its information production
profile since the total information in the network is $\sum_{i=1}^{N}\frac{\bar{H}}{N}=\bar{H}$.
Thus, ${\bf s}$ is an NE. Since this applies to any $N$, (\ref{thermeq22})
follows. Finally, since the network is always connected in any equilibrium,
then (\ref{thermeq32}) directly follows.

\end{document}